\documentclass{vldb}
\usepackage{xspace}
\usepackage{amsmath}
\usepackage{amssymb}
\usepackage{float}
\usepackage{enumitem}
\usepackage{graphicx}
\usepackage[caption=false]{subfig}
\usepackage{listings}
\usepackage{courier}
\usepackage{soul}
\usepackage[usenames,dvipsnames]{xcolor}
\usepackage{bbm, bm}
\usepackage{url}
\usepackage{todonotes}
\usepackage{verbatim}
\usepackage{algorithm}
\usepackage{array}
\usepackage{tikz}
\usepackage[pdfpagelabels=false]{hyperref}
\hypersetup{pageanchor=false}

\newlist{myitemize}{itemize}{1}
\setlist[myitemize,1]{label=\textbullet,leftmargin=10pt,itemsep=0pt,parsep=2pt}

\newcommand{\eqdef}{\stackrel{\text{def}}{=}}

\newcommand{\cut}[1]{}
\newcommand{\eat}[1]{}
\newcommand{\commentresolved}[1]{}

\newcommand{\set}[1]{\{#1\}}                    
\newcommand{\setof}[2]{\{{#1}:{#2}\}}
\usepackage{aliascnt}  		

\newtheorem{theorem}{Theorem}[section]          	
\newaliascnt{lemma}{theorem}				
\newtheorem{lemma}[lemma]{Lemma}              	
\aliascntresetthe{lemma}  					
\newaliascnt{conjecture}{theorem}			
\aliascntresetthe{conjecture}  				
\newaliascnt{remark}{theorem}				
              
\aliascntresetthe{remark}  					
\newaliascnt{corollary}{theorem}			
\newtheorem{corollary}[corollary]{Corollary}      
\aliascntresetthe{corollary}  				
\newaliascnt{definition}{theorem}			
\aliascntresetthe{definition}  				
\newaliascnt{proposition}{theorem}			
\aliascntresetthe{proposition}  				
\newaliascnt{example}{theorem}			
\newtheorem{example}[example]{Example}  	
\aliascntresetthe{example}  				

\usepackage{relsize}

\newcommand{\R}{\mathbb{R}}

\newcommand\ignore[1]{\relax}

\newcommand{\name}{EntropyDB\xspace}

\makeatletter
\def\BState{\State\hskip-\ALG@thistlm}
\newcommand{\vast}{\bBigg@{4}}
\newcommand{\Vast}{\bBigg@{5}}
\makeatother

\newcolumntype{M}[1]{>{\centering\arraybackslash}m{#1}}

\tikzstyle{mybox} = [draw=black, thick, rectangle]

\newcommand*{\mathcolor}{}
\def\mathcolor#1#{\mathcoloraux{#1}}
\newcommand*{\mathcoloraux}[3]{%
  \protect\leavevmode
  \begingroup
    \color#1{#2}#3%
  \endgroup
}

\chardef\_=`_

\newcommand{\ie}{\textrm{i.e.}\xspace}
\newcommand{\eg}{\textrm{e.g.}\xspace}
\newcommand{\E}{\mathbb{E}}
\newcommand\inner[2]{\langle #1, #2 \rangle}

\definecolor{mygreen}{rgb}{0,0.6,0}
\definecolor{mygray}{rgb}{0.5,0.5,0.5}
\definecolor{mymauve}{rgb}{0.58,0,0.82}

\lstdefinestyle{myJava}{ %
  language=java,
  mathescape=true,
  backgroundcolor=\color{white},   
  basicstyle=\scriptsize,        
  breaklines=true,                 
  captionpos=b,                    
  commentstyle=\color{mygreen},    
  escapeinside={*@}{@*},         
  keywordstyle=\color{blue},       
  stringstyle=\color{mymauve},     
}
\lstdefinestyle{mySQL}{ %
    mathescape=true,
    language=SQL,
    basicstyle=\footnotesize\ttfamily,
    deletekeywords={MIN},
    otherkeywords={LIMIT, GROUP, BY, ORDER, DESC},
    showstringspaces=false
}

\setstcolor{red}

\begin{document}

\title{Probabilistic Database Summarization \\ for Interactive Data Exploration}
\author{
\alignauthor
Laurel Orr, Magdalena Balazinska, and Dan Suciu \\
      \affaddr{University of Washington}\\
      \affaddr{Seattle, Washington, USA}\\
      \email{\{ljorr1, magda, suciu\}@cs.washington.edu}
}

\maketitle

\begin{sloppypar}
\begin{abstract}
We present a probabilistic approach to generate a small, query-able summary of a dataset for interactive data exploration. Departing from traditional summarization techniques, we use the Principle of Maximum Entropy to generate a probabilistic representation of the data that can be used to give approximate query answers.  We develop the theoretical framework and formulation of our probabilistic representation and show how to use it to answer queries. We then present solving techniques and give three critical optimizations to improve preprocessing time and query accuracy. Lastly, we experimentally evaluate our work using a 5 GB dataset of flights within the United States and a 210 GB dataset from an astronomy particle simulation. While our current work only supports linear queries, we show that our technique can successfully answer queries faster than sampling while introducing, on average, no more error than sampling and can better distinguish between rare and nonexistent values.
\end{abstract}

\section{Introduction}
\label{sec:introduction}
{\em Interactive data exploration} allows a data analyst to browse, query, transform, and visualize data at ``human speed''~\cite{crotty2016case}. It has been long recognized that general-purpose DBMSs are ill suited for interactive exploration~\cite{mozafari2015handbook}. While users require interactive responses, they do not necessarily require precise responses because either the response is used in some visualization, which has limited resolution, or an approximate result is sufficient and can be followed up with a more costly query if needed. {\em Approximate Query Processing} (AQP) refers to a set of techniques designed to allow fast but approximate answers to queries. All successful AQP systems to date rely on sampling or a combination of sampling and indexes. The sample can either be computed on-the-fly, \eg, in the highly influential work on {\em online aggregation}~\cite{hellerstein1997online} or systems like DBO~\cite{jermaine2008scalable} and Quickr~\cite{kandula2016quickr}, or precomputed offline, like in BlinkDB~\cite{agarwal2013blinkdb} or Sample$+$Seek~\cite{ding2016samplelus}. Samples have the advantage that they are easy to compute, can accurately estimate aggregate values, and are good at detecting heavy hitters. However, sampling may fail to return estimates for small populations; targeted stratified samples can alleviate this shortcoming, but stratified samples need to be precomputed to target a specific query, defeating the original purpose of AQP.

In this paper, we propose an alternative approach to interactive data exploration based on the Maximum Entropy principle (MaxEnt). The MaxEnt model has been applied in many settings beyond data exploration; \eg, the {\em multiplicative weights} mechanism~\cite{hardt2010multiplicative} is a MaxEnt model for both differentially private and, by~\cite{dwork2015generalization}, statistically valid answers to queries, and it has been shown to be theoretically optimal. In our setting of the MaxEnt model, the data is preprocessed to compute a probabilistic model. Then, queries are answered by doing probabilistic inference on this model. The model is defined as the probabilistic space that obeys some observed statistics on the data and makes no other assumptions (Occam's principle). The choice of statistics boils down to a precision/memory tradeoff: the more statistics one includes, the more precise the model and the more space required. Once computed, the MaxEnt model defines a probability distribution on possible worlds, and users can interact with this model to obtain approximate query results. Unlike a sample, which may miss rare items, the MaxEnt model can infer something about every query.

Despite its theoretical appeal, the computational challenges associated with the MaxEnt model make it difficult to use in practice. In this paper, we develop the first scalable techniques to compute and use the MaxEnt model. As an application, we illustrate it with interactive data exploration. Our first contribution is to simplify the standard MaxEnt model to a form that is appropriate for data summarization (Sec.~\ref{sec:probabilistic_approach}). We show how to simplify the MaxEnt model to be a multi-linear polynomial that has one monomial for each possible tuple (Sec.~\ref{sec:probabilistic_approach}, Eq.~(\ref{eq:p})) rather than its na\"{i}ve form that has one monomial for each possible world (Sec.~\ref{sec:background}, Eq.~(\ref{eq:pr:i})). Even with this simplification, the MaxEnt model starts by being larger than the data. For example, the flights dataset is 5 GB, but the number of possible tuples is approximately $10^{10}$, more than 5 GB. Our {\em first optimization} consists of a compression technique for the polynomial of the MaxEnt model (Sec~\ref{subsec:compress}); for example, for the flights dataset, the summary is below 200MB, while for our larger dataset of 210GB, it is less than 1GB. Our {\em second optimization} consists of a new technique for query evaluation on the MaxEnt model (Sec.~\ref{subsec:opt:aqp}) that only requires setting some variables to 0; this reduces the runtime to be on average below 500ms and always below 1s.

We find that the main bottleneck in using the MaxEnt model is computing the model itself; in other words, computing the values of the variables of the polynomial such that it matches the existing statistics over the data. Solving the MaxEnt model is difficult; prior work for multi-dimensional histograms~\cite{markl2005consistently} uses an iterative scaling algorithm for this purpose. To date, it is well understood that the MaxEnt model can be solved by reducing it to a convex optimization problem~\cite{wainwright2008GME} of a {\em dual} function (Sec.~\ref{sec:background}), which can be solved using Gradient Descent. However, even this is difficult given the size of our model. We managed to adapt a variant of Stochastic Gradient Descent called Mirror Descent~\cite{convex-optimization-algorithms-complexity}, and our optimized query evaluation technique can compute the MaxEnt model for large datasets in under a day.

In summary, in this paper, we develop the following new techniques:
\begin{myitemize}
\item A closed-form representation of the probability space of possible worlds using the Principle of Maximum Entropy, and a method
to use the representation to answer queries in expectation (Sec~\ref{sec:probabilistic_approach}).
\item Compression method for the MaxEnt summary (Sec~\ref{subsec:compress}).
\item Optimized query processing techniques (Sec~\ref{subsec:opt:aqp}).
\item A new method for selecting 2-dimensional statistics based on a modified KD-tree (Sec~\ref{subsec:stat:selection}).
\end{myitemize}

We implement the above techniques in a prototype system that we call \name and evaluate it on the flights and astronomy datasets. We find that \name can answer queries faster than sampling while introducing no more error, on average, and does better at identifying small populations.

\section{Background}
\label{sec:background}
We summarize data by fitting a probability distribution over the active domain. The distribution assumes that the domain values are distributed in a way that preserves given statistics over the data but are otherwise uniform.

For example, consider a data scientist who analyzes a dataset of flights in the United States for the month of December 2013. All she knows is that the dataset includes all flights within the 50 possible states and that there are 500,000 flights in total. She wants to know how many of those flights are from CA to NY. Without any extra information, our approach would assume all flights are equally likely and estimate that there are $500,000/50^2 = 200$ flights.

Now suppose the data scientist finds out that flights leaving CA only go to NY, FL, or WA. This changes the estimate because instead of there being $500,000/50 = 10,000$ flights leaving CA and uniformly going to all 50 states, those flights are only going to 3 states. Therefore, the estimate becomes $100,000/3 = 3,333$ flights.

This example demonstrates how our summarization technique would answer queries, and the rest of this section covers its theoretical foundation.

\subsection{Possible World Semantics}
To model a probabilistic database, we use the slotted possible world semantics where rows have an inherent unique identifier, meaning the order of the tuples matters. Our set of possible worlds is generated from the active domain and size of each relation. Each database instance is one possible world with an associated probability such that the probabilities of all possible worlds sum to one.

In contrast to typical probabilistic databases where the probability of a relation is calculated from the probability of each tuple, we calculate a relation's probability from a formula derived from the MaxEnt principle and a set of constraints on the overall distribution. This approach captures the idea that the distribution should be uniform except where otherwise specified by the given constraints.

\subsection{The Principle of Maximum Entropy}
The Principle of Maximum Entropy (MaxEnt) states that subject to prior data, the probability distribution which best represents the state of knowledge is the one that has the largest entropy. This means given our set of possible worlds, $PWD$, the probability distribution $\Pr(I)$ is one that agrees with the prior information on the data and maximizes
\begin{equation*}
-\sum_{I \in PWD}\Pr(I)\log(\Pr(I))
\end{equation*}
where $I$ is a database instance, also called possible world. The above probability must be normalized, $\sum_I \Pr(I)=1$, and must satisfy the prior information represented by a set of $k$ expected value constraints:
\begin{equation}
\label{eq:e:s}
s_{j} = \E[\phi_{j}(I)], \ \ j=1,k 
\end{equation}
where $s_{j}$ is a known value and $\phi_{j}$ is a function on $I$ that returns a numerical value in $\mathbb{R}$.
One example constraint is that the number of flights from CA to WI is 0.

Following prior work on the MaxEnt principle and solving constrained optimization problems \cite{berger1996nlpapproach,wainwright2008GME,re2012understanding}, the MaxEnt probability distribution takes the form
\begin{equation}
\label{eq:pr:i}
\Pr(I) = \frac{1}{Z}\exp\left(\sum_{j = 1}^{k}\theta_{j}\phi_{j}(I)\right) 
\end{equation}
where $\theta_{j}$ is a parameter and $Z$ is the following normalization constant:
\begin{equation*}
Z \eqdef \sum_{I \in PWD}\left(\exp\left(\sum_{j = 1}^{k}\theta_{j}\phi_{j}(I)\right)\right).
\end{equation*}

To compute the $k$ parameters $\theta_j$, we must solve the non-linear system of $k$ equations, Eq.~(\ref{eq:e:s}), which is computationally difficult.  However, it turns out \cite{wainwright2008GME} that Eq.~(\ref{eq:e:s}) is equivalent to $\partial \Psi / \partial \theta_j = 0$ where the {\em dual} $\Psi$ is defined as:
\begin{equation*}
\Psi \eqdef \sum_{j = 1}^{k} s_{j}\theta_{j} - \ln\left(Z\right).
\end{equation*}
Furthermore, $\Psi$ is concave, which means solving for the $k$ parameters can be achieved by maximizing $\Psi$. We note that $Z$ is called the {\em partition function}, and its log, $\ln(Z)$, is called the {\em cumulant}.

\section{\name Approach}
\label{sec:probabilistic_approach}
This section explains how we use the MaxEnt model for approximate query answering. We first show how we use the MaxEnt framework to transform a single relation $R$ into a probability distribution represented by $P$. We then explain how we use $P$ to answer queries over $R$.

\subsection{Maximum Entropy Model of Data}
\label{sec:theoretical_setup}

We consider a single relation with $m$ attributes and schema $R(A_1,\ldots, A_m)$ where each attribute, $A_i$, has an active domain $D_i$, assumed to be discrete and ordered.\footnote{We support continuous data types by bucketizing their active domains.}  Let $Tup = D_1 \times D_2 \times \dots \times D_m = \{t_1, \ldots, t_d\}$ be the set of all possible tuples.  Denoting  $N_i = |D_i|$, we have $d = |Tup| = \prod_{i=1}^m |N_i|$.

An {\em instance} for $R$ is an ordered bag of $n$ tuples, denoted $I$. For each $I$, we form a frequency vector which is a $d$-dimensional vector\footnote{This is a standard data model in several applications, such as differential privacy \cite{li2010optimizing}.} $\mathbf{n}^I = [n^I_1,\ldots, n^I_d] \in \R^d$, where each number $n^I_i$ represents the count of the tuple $t_i \in Tup$ in $I$ (Fig.~\ref{fig:model}).  The mapping from $I$ to $\mathbf{n}^I$ is not one-to-one because the instance $I$ is ordered, and two distinct instances may have the same counts. Further, for any instance $I$ of cardinality $n$, $||\mathbf{n}^I||_1 = \sum_i n^I_i = n$. The frequency vector of an instance consisting of a single tuple $\set{t_i}$ is denoted $\mathbf{n}^{t_i} = [0,\ldots,0,1,0,\ldots,0]$ with a single value $1$ in the $i$th position; \ie, $\setof{\mathbf{n}^{t_i}}{i = 1,d}$ forms a basis for $\R^d$.

While the MaxEnt principle allows us, theoretically, to answer any query probabilistically by averaging the query over all possible instances; in this paper, we limit our discussion to linear queries. A {\em linear query} is a $d$-dimensional vector $\mathbf{q} = [q_1, \ldots, q_d]$ in $\R^d$.  The answer to $\mathbf{q}$ on instance $I$ is the dot product $\inner{\mathbf{q}}{\mathbf{n}^I} = \sum_{i=1}^d q_i n^I_i$. With some abuse of notation, we will write $\mathbf{I}$ when referring to $\mathbf{n}^I$ and $\mathbf{t}_i$ when referring to $\mathbf{n}^{t_i}$.  Notice that $\inner{\mathbf{q}}{\mathbf{t}_i} = q_i$, and, for any instance $I$, $\inner{\mathbf{q}}{\mathbf{I}} = \sum_i n_i^I \inner{\mathbf{q}}{\mathbf{t}_i}$.

Fig.~\ref{fig:model} illustrates the data and query model. Any counting query is a vector $\mathbf{q}$ where all coordinates are 0 or 1 and can be equivalently defined by a predicate $\pi$ such that $\inner{\mathbf{q}}{\mathbf{I}} = |\sigma_\pi(I)|$; with more abuse, we will use $\pi$ instead of $\mathbf{q}$ when referring to a counting query.  Other SQL queries can be modeled using linear queries, too. For example, \texttt{SELECT A, COUNT($*$) AS cnt FROM R GROUP BY A ORDER BY cnt DESC LIMIT 10} corresponds to several linear queries, one for each group, where the outputs are sorted and the top 10 returned.

\begin{figure}
\scriptsize
\begin{tikzpicture}
\node [mybox,minimum width=8.35cm] (box1) at (0, 0){
\begin{minipage}{0.45\textwidth}
Domains:
\vspace{-4pt}
\begin{align*}
  &D_1 = \set{a_1,a_2}  &N_1 = 2\\
  &D_2 = \set{b_1,b_2}  &N_2 = 2\\
  &Tup = \set{(a_1,b_1),(a_1,b_2),(a_2,b_1),(a_2,b_2)}  &d = 4
\end{align*}
\end{minipage}
};
\node [mybox,minimum width=8.35cm] (box2) at (0, -1.9) {
\begin{minipage}[t]{0.22\textwidth}
Database Instance:
\vspace{2pt}
\newline
\begin{tabular}{l|l|l|} \cline{2-3}
$I$:  & $A$ & $B$ \\ \cline{2-3}
1 & $a_1$ & $b_1$ \\
2 & $a_1$ & $b_2$ \\
3 & $a_2$ & $b_2$ \\
4 & $a_1$ & $b_1$ \\
5 & $a_2$ & $b_2$ \\ \cline{2-3}
\end{tabular}
\end{minipage}
\begin{minipage}[t]{0.22\textwidth}
Query: 
\vspace{6pt}
\newline
\begin{tabular}{rl}
\texttt{q:} & \texttt{SELECT COUNT(*)} \\
            & \texttt{FROM R} \\
            & \texttt{WHERE A = a1}  
\end{tabular}
\end{minipage}
};
\node [mybox,minimum width=8.35cm] (box3) at (0, -3.5){
\begin{minipage}{0.45\textwidth}
Modeling Data and Query: n = 5,\ m = 2
\begin{align*}
& \mathbf{n}^{\mathbf{I}} = (2,1,0,2) ~~ \mathbf{q} = (1,1,0,0)  ~~
  \inner{\mathbf{q}}{\mathbf{n}^{\mathbf{I}}} = 3
  \mbox{\hspace{6pt}also denoted } \inner{\mathbf{q}}{\mathbf{I}} \mbox{\hspace{0.3cm}}
\end{align*}
\end{minipage}
};
\end{tikzpicture}
\caption{Illustration of the data and query model}
\label{fig:model}
\end{figure}

Our goal is to compute a summary of the data that is small yet allows us to approximatively compute the answer to any linear query.  We assume that the cardinality $n$ of $R$ is fixed and known.  In addition, we know $k$ statistics, $\Phi = \setof{(\mathbf{c}_j,s_j)}{j=1,k}$, where $\mathbf{c}_j$ is a linear query and $s_j \geq 0$ is a number.  Intuitively, the statistic $(\mathbf{c}_j,s_j)$ asserts that $\inner{\mathbf{c}_j}{I} = s_j$. For example, we can write 1-dimensional and 2-dimensional (2D) statistics like $|\sigma_{A_1 = 63}(I)| = 20$ and $|\sigma_{A_1 \in [50,99]\wedge A_2 \in [1,9]}(I)| = 300$.

Next, we derive the MaxEnt distribution for the possible instances $I$ of a fixed size $n$. We replace the exponential parameters $\theta_j$ with $\ln(\alpha_j)$ so that Eq.~(\ref{eq:pr:i}) becomes

\begin{equation}
\label{eq:pr:n}
\Pr(I) = \frac{1}{Z}\prod_{j=1,k} \alpha_j^{\inner{\mathbf{c}_j}{\mathbf{I}}}.
\end{equation}

We prove the following about the structure of the partition function $Z$:
\begin{lemma}
The partition function is given by
\begin{equation}
\label{eq:z:p}
Z = P^n
\end{equation}
where $P$ is the multi-linear polynomial
\begin{equation}
\label{eq:p}
P(\alpha_1, \ldots, \alpha_k) \eqdef \sum_{i=1,d} \prod_{j=1,k} \alpha_j^{\inner{\mathbf{c}_j}{\mathbf{t}_i}}.
\end{equation}
\end{lemma}

\begin{proof}
Fix any $\mathbf{n} = [n_1,\ldots, n_d]$ such that $||\mathbf{n}||_1 = \sum_{i=1}^d n_i = n$.  The number of instances $I$ of cardinality $n$ with $\mathbf{n}^I = \mathbf{n}$ is $n!/\prod_i n_i!$. Furthermore, for each such instance, $\inner{\mathbf{c}_j}{\mathbf{I}} = \inner{\mathbf{c}_j}{\mathbf{n}} = \sum_i n_i\inner{\mathbf{c}_j}{\mathbf{t}_i}$.  Therefore,

\begin{align*}
Z = & \sum_I \Pr(I) = \sum_{\mathbf{n}: ||\mathbf{n}||_1 = n} \frac{n!}{\prod_{i} n_i!} \prod_{j=1,k} \alpha_j^{\sum_{i}n_i \inner{\mathbf{c}_j}{\mathbf{t}_i}} \\
= & \left(\sum_{i=1,d} \prod_{j=1,k} \alpha_j^{\inner{\mathbf{c}_j}{\mathbf{t}_i}}\right)^n = P^n.
\end{align*}
\end{proof}

The {\em data summary} consists of the polynomial $P$ (Eq.~(\ref{eq:p})) and the values of its parameters $\alpha_j$; the polynomial is defined by the  linear queries $\mathbf{c}_j$ in the statistics $\Phi$, and the parameters are computed from the numerical values $s_j$.

\begin{example}
\label{ex:simple}
Consider a relation with three attributes $R(A, B, C)$, and assume that the domain of each attribute has 2 distinct elements. Assume $n = 10$ and the only statistics in $\Phi$ are the following 1-dimensional statistics:
\begin{equation*}
\begin{array}{lll}
(A = a_1,\ 3) & (B = b_1,\ 8) & (C = c_1,\ 6) \\
(A = a_2,\ 7) & (B = b_2,\ 2) & (C = c_2,\ 4). \\
\end{array}
\end{equation*}
The first statistic asserts that $|\sigma_{A=a_1}(I)|=3$, etc. The polynomial $P$ is
\begin{align*}
P = 
&\alpha_{1}\beta_{1}\gamma_{1} + \alpha_{1}\beta_{1}\gamma_{2} + 
\alpha_{1}\beta_{2}\gamma_{1} + \alpha_{1}\beta_{2}\gamma_{2} + \\
&\alpha_{2}\beta_{1}\gamma_{1} + \alpha_{2}\beta_{1}\gamma_{2} + 
\alpha_{2}\beta_{2}\gamma_{1} + \alpha_{2}\beta_{2}\gamma_{2}
\end{align*}
where $\alpha_1,\alpha_2$ are variables associated with the statistics on $A$, $\beta_1,\beta_2$ are for $B$\footnote{We abuse notation here for readability. Technically, $\alpha_{i} = \alpha_{a_i}$, $\beta_{i} = \alpha_{b_i}$, and $\gamma_{i} = \alpha_{c_i}$.}, and $\gamma_1,\gamma_2$ are for $C$.

Consider the concrete instance
\begin{equation*}
I = \set{(a_1,b_1,c_1), (a_1, b_2,c_2), \ldots, (a_1, b_2, c_2)}
\end{equation*}
where the tuple $(a_1,b_2,c_2)$ occurs 9 times. Then, $\Pr(I) = \alpha_1^{10}\beta_1\beta_2^9\gamma_1\gamma_2^9/P^{10}$.
\end{example}

\begin{example}
\label{ex:complex1}
Continuing the previous example, we add the following multi-dimensional statistics to $\Phi$:
\begin{equation*}
\begin{array}{ll}
(A = a_1 \land B = b_1,\ 2) & (B = b_1 \land C = c_1,\ 5) \\
(A = a_2 \land B = b_2,\ 1) & (B = b_2 \land C = c_1,\ 1). \\
\end{array}
\end{equation*}
$P$ is now
\begin{align}
\label{eq:pnew}
P = 
&\alpha_{1}\beta_{1}\gamma_{1}\mathcolor{red}{[\alpha\beta]_{1,1}}\mathcolor{red}{[\beta\gamma]_{1,1}} + \alpha_{1}\beta_{1}\gamma_{2}\mathcolor{red}{[\alpha\beta]_{1,1}} + \nonumber \\
&\alpha_{1}\beta_{2}\gamma_{1}\mathcolor{red}{[\beta\gamma]_{2,1}} + \alpha_{1}\beta_{2}\gamma_{2} + \nonumber \\
&\alpha_{2}\beta_{1}\gamma_{1}\mathcolor{red}{[\beta\gamma]_{1,1}} + \alpha_{2}\beta_{1}\gamma_{2} + \nonumber \\
&\alpha_{2}\beta_{2}\gamma_{1}\mathcolor{red}{[\alpha\beta]_{2,2}}\mathcolor{red}{[\beta\gamma]_{2,1}} + \alpha_{2}\beta_{2}\gamma_{2}\mathcolor{red}{[\alpha\beta]_{2,2}}.
\end{align}
The red variables are the added 2-dimensional statistic variables; we use $[\alpha\beta]_{1,1}$ to denote a {\em single} variable corresponding to a 2D statistics on the attributes $AB$. Notice that each  red variable only occurs with its related 1-dimensional variables. $\alpha\beta_{1,1}$, for example, is only in the same term as $\alpha_{1}$ and $\beta_{1}$.  

Now consider the earlier instance $I$. Its probability becomes $\Pr(I) = \alpha_1^{10}\beta_1\beta_2^9\gamma_1\gamma_2^9[\alpha\beta]_{1,1}[\beta\gamma]_{1,1}/P^{10}$.
\end{example}

To facilitate analytical queries, we choose the set of statistics $\Phi$ as follows:
\begin{myitemize}
\item Each statistic $\phi_j=(\mathbf{c}_j,s_j)$ is associated with some predicate $\pi_j$ such that $\inner{\mathbf{c}_j}{\mathbf{I}} = |\sigma_{\pi_j}(I)|$. It follows that for every tuple $t_i$, $\inner{\mathbf{c}_j}{\mathbf{t}_i}$ is either 0 or 1; therefore, each variable $\alpha_j$ has degree 1 in the polynomial $P$ in Eq.~(\ref{eq:p}). 

\item For each domain $D_i$, we include a complete set of 1-dimensional statistics in our summary. In other words, for each $v \in D_i$, $\Phi$ contains one statistic with predicate $A_i = v$.  We denote $J_i \subseteq [k]$ the set of indices of the 1-dimensional statistics associated with $D_i$; therefore, $|J_i| = |D_i| = N_i$.

\item We allow multi-dimensional statistics to be given by arbitrary predicates. They may be overlapping and/or incomplete; \eg, one statistic may count the tuples satisfying $A_1  \in [10,30] \wedge A_2 = 5$ and another count the tuples satisfying $A_2 \in [20, 40] \wedge A_4=20$. 

\item We assume the number of 1-dimensional statistics dominates the number of attribute combinations; i.e., $\sum_{i=1}^m N_i \gg 2^m$.


\item If some domain $D_i$ is large, it is beneficial to reduce the
  size of the domain using equi-width buckets. In that case, we assume the elements of $D_i$ represent buckets, and $N_i$ is the number of buckets.

\item We enforce our MaxEnt distribution to be {\em overcomplete}~\cite[pp.40]{wainwright2008GME} (as opposed to {\em minimal}). More precisely, for any attribute $A_i$ and any instance $I$, we have $\sum_{j \in J_i} \inner{\mathbf{c}_j}{\mathbf{I}} = n$, which means that some statistics are redundant since they can be computed from the others and from the size of the instance $n$.
\end{myitemize}

Note that as a consequence of overcompleteness, for any attribute $A_i$, one can write $P$ as a linear expression

\begin{equation}
\label{eq:p1i}
P = \sum_{j \in J_i} \alpha_j P_j
\end{equation}
where each $P_j$, $j \in J_i$ is a polynomial that does not contain the variables $(\alpha_j)_{j \in J_i}$. In Example~\ref{ex:complex1}, the 1-dimensional variables for $A$ are $\alpha_1$, $\alpha_2$, and indeed, each monomial in Eq.~(\ref{eq:pnew}) contains exactly one of these variables. One can write $P$ as $P = \alpha_1 P_1 + \alpha_2 P_2$ where $\alpha_1 P_1$ represents the first two lines and $\alpha_2 P_2$ represents the last two lines in Eq.~(\ref{eq:pnew}). $P$ is also linear in $\beta_1$, $\beta_2$ and in $\gamma_1$, $\gamma_2$.

\subsection{Query Answering}
\label{subsec:query:answer}
In this section, we show how to use the data summary to approximately answer a linear query $q$ by returning its expected value $\E[\inner{\mathbf{q}}{I}]$. The summary (the polynomial $P$ and the values of its variables $\alpha_j$) uniquely define a probability space on the possible worlds (Eq.~(\ref{eq:pr:n}) and (\ref{eq:p})). We start with a well known result in the MaxEnt model. If $\mathbf{c}_\ell$ is the linear query associated with the variable $\alpha_\ell$, then

\begin{equation}
\label{eq:ec}
  \E[\inner{\mathbf{c}_\ell}{\mathbf{I}}] = \frac{n}{P}\frac{\alpha_\ell\partial P}{\partial \alpha_\ell}.
\end{equation}
We review the proof here. The expected value of $\inner{\mathbf{c}_\ell}{\mathbf{I}}$ over the probability space (Eq.~(\ref{eq:pr:n})) is

\begin{align*}
  \E[\inner{\mathbf{c}_\ell}{\mathbf{I}}] = &\frac{1}{P^n} \sum_{\mathbf{I}} \inner{\mathbf{c}_\ell}{\mathbf{I}}\prod_j \alpha_j^{\inner{\mathbf{c}_j}{\mathbf{I}}}
 =  \frac{1}{P^n} \sum_{\mathbf{I}} \frac{\alpha_\ell \partial}{\partial \alpha_\ell} \prod_j \alpha_j^{\inner{\mathbf{c}_j}{\mathbf{I}}} \\
 = & \frac{1}{P^n} \frac{\alpha_\ell \partial}{\partial \alpha_\ell} \sum_{\mathbf{I}} \prod_j \alpha_j^{\inner{\mathbf{c}_j}{\mathbf{I}}}
 =  \frac{1}{P^n} \frac{\alpha_\ell \partial P^n}{\partial \alpha_\ell} = \frac{n}{P}\frac{\alpha_\ell\partial P}{\partial \alpha_\ell}.
\end{align*}

To compute a new linear query $\mathbf{q}$, we add it to the statistical queries $c_j$, associate it with a fresh variable $\beta$,
and denote $P_{\mathbf{q}}$ the extended polynomial:

\begin{align}
  P_{\mathbf{q}}(\alpha_1, \ldots, \alpha_k,\beta) \eqdef \sum_{i=1,d} \prod_{j=1,k} \alpha_j^{\inner{\mathbf{c}_j}{\mathbf{t}_i}}\beta^{\inner{\mathbf{q}}{\mathbf{t}_i}}\label{eq:pq}
\end{align}

Notice that $P_{\mathbf{q}}[\beta=1] \equiv P$; therefore, the extended data summary defines the same probability space as $P$. We can apply Eq.~(\ref{eq:ec}) to the query $\mathbf{q}$ to derive:

\begin{equation}
\label{eq:eq}
  \E[\inner{\mathbf{q}}{\mathbf{I}}] = \frac{n}{P}\frac{\partial P_{\mathbf{q}}}{\partial \beta}.
\end{equation}

This leads to the following na\"{i}ve strategy for computing the expected value of $\mathbf{q}$: extend $P$ to obtain $P_{\mathbf{q}}$ and apply formula Eq.~(\ref{eq:eq}). One way to obtain $P_{\mathbf{q}}$ is to iterate over all monomials in $P$ and add $\beta$ to the monomials corresponding to tuples counted by $\mathbf{q}$. As this is inefficient, Sec.~\ref{subsec:opt:aqp} describes how to avoid modifying the polynomial altogether.

\subsection{Probabilistic Model Computation}
\label{subsec:solving}
We now describe how to compute the parameters of the summary. Given the statistics $\Phi = \setof{(\mathbf{c}_j,s_j)}{j=1,k}$, we need to find values of the variables $\setof{\alpha_j}{j=1,k}$ such that $\E[\inner{\mathbf{c}_j}{\mathbf{I}}] = s_j$ for all $j=1,k$. As explained in Sec~\ref{sec:background}, this is equivalent to maximizing the dual function $\Psi$:

\begin{equation}
\label{eq:dual}
\Psi \eqdef \sum_{j=1}^k s_j \ln(\alpha_j) - n \ln P.
\end{equation}

Indeed, maximizing $P$ reduces to solving the equations $\partial \Psi/\partial \alpha_j = 0$ for all $j$. Direct calculation gives us $\partial \Psi/\partial \alpha_j = \frac{s_j}{\alpha_j} - \frac{n}{P}\frac{\partial P}{\partial \alpha_j} = 0$, which is equivalent to $s_j - \E[\inner{\mathbf{c}_j}{\mathbf{I}}]$ by Eq.~(\ref{eq:ec}). The dual function $\Psi$ is concave, and hence it has a single maximum value that can be obtained using convex optimization techniques such as Gradient Descent.

In particular, we achieve fastest convergence rates using a variant of Stochastic Gradient Descent (SGD) called Mirror Descent~\cite{convex-optimization-algorithms-complexity},  where each iteration chooses some $j=1,k$ and updates $\alpha_j$ by solving $\frac{n \alpha_j}{P}\frac{\partial P}{\partial \alpha_j} = s_j$ while keeping all other parameters fixed. In other words, the step of SGD is chosen to solve $\partial \Psi/\partial \alpha_j = 0$. Denoting $P_{\alpha_{j}} \eqdef \frac{\partial P}{\partial \alpha_j}$ and solving, we obtain:
\begin{equation}
\label{eq:update_step} 
\alpha_{j} = \frac{s_j(P - {\alpha_{j}}P_{\alpha_{j}})}{(n-s_j)P_{\alpha_{j}}}.
\end{equation}
Since $P$ is linear in each $\alpha$, neither $P - {\alpha_{j}}P_{\alpha_{j}}$ nor $P_{\alpha_{j}}$ contain any $\alpha_{j}$ variables.

We repeat this for all $j$, and continue this process until all
differences $|s_j - \frac{n\alpha_{j}P_{\alpha_{j}}}{P}|$, $j=1,k$,
are below some threshold. Algorithm~\ref{alg:solve} shows pseudocode for the solving process.
\begin{small}
\begin{algorithm}[t]
\caption{Solving for the $\alpha$s}
\label{alg:solve}
\begin{lstlisting}[style=myJava] 
maxError = infinity
while maxError >= threshold do
  maxError = -1
  for each alpha do
      value = *@$\frac{s_j(P - {\alpha_{j}}P_{\alpha_{j}})}{(n-s_j)P_{\alpha_{j}}}$@*
      alpha = value
      error = *@$value - \frac{n\alpha_{j}P_{\alpha_{j}}}{P}$@*
      maxError = max(error, maxError)
\end{lstlisting}
\end{algorithm}
\end{small}

\section{Optimizations}
\label{sec:optimizations}
We now discuss three optimizations: (1) summary compression in Sec.~\ref{subsec:compress}, (2) optimized query processing in Sec.~\ref{subsec:opt:aqp}, and (3) selection of statistics in Sec.~\ref{subsec:stat:selection}.

\subsection{Compression of the Data Summary}
\label{subsec:compress}

The summary consists of the polynomial $P$ that, by definition, has $|Tup|$ monomials where $|Tup| = \prod_{i=1}^m N_i$. We describe a technique that compresses the summary to a size closer to $O(\sum_i N_i)$.

We start by walking through an example with three attributes, $A$, $B$, and $C$, each with an active domain of size $N_1=N_2=N_3=1000$. Suppose first that we have only 1D statistics.  Then, instead of representing $P$ as a sum of $1000^3$ monomials, $P = \sum_{i,j,k \in [1000]} \alpha_i\beta_j \gamma_k$, we factorize it to $P = (\sum \alpha_i)(\sum \beta_j)(\sum \gamma_k)$; the new representation has size $3 \cdot 1000$.

Now, suppose we add a single 3D statistic on $ABC$: $A = 3 \wedge B = 4 \wedge C=5$.  The new variable, call it $\delta$, occurs in a single monomial of $P$, namely $\alpha_3\beta_4\gamma_5\delta$.  Thus, we can compress $P$ to $(\sum \alpha_i)(\sum \beta_j)(\sum \gamma_k) + \alpha_3\beta_4\gamma_5(\delta-1)$.

Instead, suppose we add a single 2D range statistics on $AB$, say $A \in [101, 200] \wedge B \in [501, 600]$, and call its associated variable $\delta_1$.  This will affect $100 \cdot 100 \cdot 1000$ monomials.  We can avoid enumerating them by noting that they, too, factorize.  The polynomial compresses to $(\sum \alpha_i)(\sum \beta_j)(\sum \gamma_k) + (\sum_{i=101}^{200}\alpha_i)(\sum_{j=501}^{600} \beta_j)(\sum \gamma_k)(\delta_1-1)$.

Finally, suppose we have three 2D statistics: the previous one on $AB$ plus the statistics $B \in [551, 650] \wedge C \in [801, 900]$ and $B \in [651, 700] \wedge C \in [701, 800]$ on $BC$.  Their associated variables are $\delta_1$, $\delta_2$, and $\delta_3$.  Now we need to account for the fact that $100\cdot 50 \cdot 100$ monomials contain both $\delta_1$ and $\delta_2$.  Applying the inclusion/exclusion principle, $P$ compresses to the following (the \textbf{i} and \textbf{ii} labels are referenced later).

\begin{small}
  \begin{align}
  \label{eq:ex:p}
   P= & \overbrace{(\sum \alpha_i)(\sum \beta_j)(\sum \gamma_k)}^{(\textbf{i})}
    + \overbrace{(\sum \gamma_k)}^{(\textbf{i})}\overbrace{(\sum_{101}^{200}\alpha_i)(\sum_{501}^{600} \beta_j)(\delta_1-1)}^{(\textbf{ii})}\\
    + &\overbrace{(\sum \alpha_i)}^{(\textbf{i})} \overbrace{\left[(\sum_{551}^{650} \beta_j)(\sum_{801}^{900} \gamma_k)(\delta_2-1)+(\sum_{651}^{700} \beta_j)(\sum_{701}^{800} \gamma_k)(\delta_3-1)\right]}^{(\textbf{ii})}\\
    + &\overbrace{(\sum_{101}^{200} \alpha_i)(\sum_{551}^{600} \beta_j)(\sum_{801}^{900} \gamma_k)(\delta_1-1)(\delta_2-1)}^{(\textbf{ii})}.
  \end{align}
\end{small}

\noindent The size, counting only the $\alpha$s, $\beta$s, and $\gamma$s for simplicity, is $3000 + 1200 + 1350 + 250$.

Before proving the general formula for $P$, note that this compression is related to standard algebraic factorization techniques involving kernel extraction and rectangle coverings \cite{hosangadi2004factoringae}; both techniques reduce the size of a polynomial by factoring out divisors. The standard techniques, however, are unsuitable for our use because they require enumeration of the product terms in the sum-of-product (SOP) polynomial to extract kernels and form cube matrices. Our polynomial in SOP form is too large to be materialized, making these techniques infeasible. It is future work to investigate other factorization techniques geared towards massive polynomials.

We now make the following three assumptions for the rest of the paper.
\begin{myitemize}
  \item Each predicate has the form $\pi_j = \bigwedge_{i=1}^m \rho_{ij}$ where $m$ is the number of attributes and $\rho_{ij}$ is the projection of $\pi_j$ onto $A_i$. If $j \in J_i$, then $\pi_j \equiv \rho_{ij}$. For any set of indices of multi-dimensional statistics $S \subset [k]$, we denote $\rho_{iS} \eqdef \bigwedge_{j \in S} \rho_{ij}$, and $\pi_S \eqdef \bigwedge_i \rho_{iS}$; as usual, when $S =\emptyset$, then $\rho_{i\emptyset} \equiv \texttt{true}$.

  \item Each $\rho_{ij}$ is a range predicate $A_i \in [u,v]$.

  \item For each $\mathcal{I}$, the multi-dimensional statistics whose attributes are exactly those in $\mathcal{I}$ are disjoint; i.e., for $j_1$, $j_2$ whose attributes are $\mathcal{I}$, $\rho_{ij} \not\equiv \texttt{true}$ for $i \in \mathcal{I}$, $\rho_{ij} \equiv \texttt{true}$ for $i \not\in \mathcal{I}$, and $\pi_{j_1} \wedge \pi_{j_2} \equiv \texttt{false}$.
\end{myitemize}

Using this, define $J_{\mathcal{I}} \subseteq \mathcal{P}([k])$\footnote{$\mathcal{P}([k])$ is the power set of $\set{1, 2, \ldots, k}$} for $\mathcal{I} \subseteq [m]$ to be the set of sets of multi-dimensional statistics whose {\em combined} attributes are $\setof{A_i}{i \in \mathcal{I}}$ and whose intersection is non-empty (i.e., not \texttt{false}). In other words, for each $S \in J_{\mathcal{I}}$, $\rho_{iS} \not\in \{\texttt{true}, \texttt{false}\}$ for $i \in \mathcal{I}$ and $\rho_{iS} \equiv \texttt{true}$ for $i \notin \mathcal{I}$.

For example, suppose we have the three 2D statistics from before: $\pi_{j_1} = A_1 \in [101, 200] \wedge A_2 \in [501, 600]$, $\pi_{j_2} = A_2 \in [551, 650] \wedge A_3 \in [801, 900]$, and $\pi_{j_3} = A_2 \in [651, 700] \wedge A_3 \in [701, 800]$. Then, $\{j_1\} \in J_{\{1, 2\}}$ and $\{j_2\}, \{j_3\} \in J_{\{2, 3\}}$. Further, $\{j_1, j_2\} \in J_{\{1, 2, 3\}}$ because $\rho_{2j_1} \wedge \rho_{2j_2} \not\equiv \texttt{false}$. However, $\{j_1, j_3\} \notin J_{\{1, 2, 3\}}$ because $\rho_{2j_1} \wedge \rho_{2j_3} \equiv \texttt{false}$. Using these definitions, we now give the compression.

\begin{theorem} \label{th:compress} The polynomial $P$ is equivalent to:
\begin{small}
  \begin{align*}
    P = &\sum_{\mathcal{I} \subseteq [m]}\overbrace{\left(\prod_{i \notin \mathcal{I}} \sum_{j \in J_i} \alpha_j \right)}^{(\textbf{i})} \\
    &\underbrace{\left(\sum_{S \in J_{\mathcal{I}}}\left(\prod_{i \in \mathcal{I}} \sum_{\substack{j \in J_i \\ \pi_j \land \rho_{iS} \not\equiv \texttt{false}}} \alpha_j \right)\left(\prod_{j \in S}(\alpha_j - 1) \right) \right)}_{(\textbf{ii})}
  \end{align*}
\end{small}
\end{theorem}
The proof uses induction on the size of $\mathcal{I}$, but we omit it for lack of space.

To give intuition, when $\mathcal{I} = \emptyset$, we get the sum over the 1D statistics because when $S = \emptyset$, $(\textbf{ii})$ equals 1. When $\mathcal{I}$ is not empty, $(\textbf{ii})$ has one summand for each set $S$ of multi-dimensional statistics whose attributes are $\mathcal{I}$ and whose intersection is non-empty. For each such $S$, the summand sums up all 1-dimensional variables $\alpha_j$, $j\in J_i$ that are in the $i$th projection of the predicate $\pi_S$ (this is what the condition $(\pi_j \wedge \rho_{iS})\not\equiv\texttt{false}$ checks) and multiplies with terms $\alpha_j-1$ for $j \in S$.

At a high level, our algorithm computes the compressed representation of $P$ by first computing the summand for when $I = \emptyset$ by iterating over all 1-dimensional statistics. It then iterates over the multi-dimensional statistics, and builds a map from $\mathcal{I}$ to the attributes that are defined on $\mathcal{I}$; i.e., $\mathcal{I} \rightarrow J_{\mathcal{I}}$ such that $|S| = 1$ for $S \in J_{\mathcal{I}}$. It then iteratively loops over this map, taking the cross product of different values, $J_{\mathcal{I}}$ and $J_{\mathcal{I'}}$, to see if any new $J_{\mathcal{I \cup I'}}$ can be generated. If so, $J_{\mathcal{I \cup I'}}$ is added to the map. Once done, it iterates over the keys in this map to build the summands for each $\mathcal{I}$.

The algorithm can be used during query answering to compute the compressed representation of $P_{\mathbf{q}}$ from $P$ (Sec.~\ref{subsec:query:answer}) by rebuilding \textbf{ii} for the new $\mathbf{q}$. However, as this is inefficient and may increase the size of our polynomial, our system performs query answering differently, as explained in the next section.

We now analyze the size of the compressed polynomial $P$. Let $B_a$ denote the number of non-empty $J_{\mathcal{I}}$; i.e., the number of unique multi-dimensional attribute sets. Since $B_a < 2^m$ and $\sum_{i=1}^m N_i \gg 2^m$, $B_a$ is dominated by $\sum_{i=1}^m N_i$. For some $\mathcal{I}$, part $(\textbf{i})$ of the compression is $O(\sum_{i=1}^m N_i)$. Part $(\textbf{ii})$ of the compression is more complex. For some $S \in J_{\mathcal{I}}$, the summand is of size $O(\sum_{i=1}^m N_i + |S|)$. As $|S| \leq B_a \ll \sum_{i=1}^m N_i$, the summand is only $O(\sum_{i=1}^m N_i)$. Putting it together, for some $\mathcal{I}$, we have the size is $O(\sum_{i=1}^m N_i + |J_{\mathcal{I}}|\sum_{i=1}^m N_i) = O(|J_{\mathcal{I}}|\sum_{i=1}^m N_i)$.

$|J_{\mathcal{I}}|$ is the number of sets of multi-dimensional statistics whose {\em combined} attributes are $\setof{A_i}{i \in \mathcal{I}}$ and whose intersection is non-empty. A way to think about this is that each $A_i$ defines a dimension in $|\mathcal{I}|$-dimensional space. Each $S \in J_{\mathcal{I}}$ defines a rectangle in this hyper-space. This means $|J_{\mathcal{I}}|$ is the number of rectangle coverings defined by the statistics over $\setof{A_i}{i \in \mathcal{I}}$. If we denote $R = \max_{\mathcal{I}}|J_{\mathcal{I}}|$, then the size of the summand is $O(R\sum_{i=1}^m N_i)$.

Further, although there are $2^m$ possible $\mathcal{I}$, $J_{\mathcal{I}}$ is non-empty for only $B_a + 1$ $\mathcal{I}$ (the 1 is from $\mathcal{I} = \emptyset$). Therefore, the size of the compression is $O(B_a R \sum_{i=1}^m N_i)$.

\begin{theorem}
The size of the polynomial is $O(B_a R\sum_{i=1}^m N_i)$ where $B_a$ is the number of unique multi-dimensional attribute sets and $R$ is the largest number of rectangle coverings defined by the statistics over some $\mathcal{I}$.
\end{theorem}

In the worst case, if one gathers all possible multi-dimensional statistics, this compression will be worse than the uncompressed polynomial, which is of size $O(\prod_{i = 1}^m N_i)$. However, in practice, $B_a < m$, and $R$ is dependent on the number and type of statistics collected and results in a significant reduction of polynomial size to one closer to $O(\sum_{i=1}^m N_i)$ (see Fig.~\ref{fig:kd_methodcompare} discussion).

\subsection{Optimized Query Answering}
\label{subsec:opt:aqp}
In this section, we assume that the query $\mathbf{q}$ is a counting query defined by a conjunction of predicates, one over each attribute $A_i$; \ie, $\mathbf{q} = |\sigma_\pi(I)|$, where
\begin{equation}
\label{eq:piq}
\pi = \rho_1 \wedge \cdots \wedge \rho_m 
\end{equation}

and $\rho_i$ is a predicate over the attribute $A_i$. If $\mathbf{q}$ ignores $A_i$, then we simply set $\rho_i \equiv \texttt{true}$.  Our goal is to compute $\E[\inner{\mathbf{q}}{\mathbf{I}}]$. In Sec.~\ref{subsec:query:answer}, we described a direct approach that consists of constructing a new polynomial $P_{\mathbf{q}}$ and returning Eq.~(\ref{eq:eq}). However, as described in Sec.~\ref{subsec:query:answer} and Sec.~\ref{subsec:compress}, this may be expensive.

We describe here an optimized approach to compute $\E[\inner{\mathbf{q}}{\mathbf{I}}]$ directly from $P$. The advantage of this method is that it does not require any restructuring or rebuilding of the polynomial. Instead, it can use any optimized oracle for evaluating $P$ on given inputs. Our optimization has two parts: a new formula $\E[\inner{\mathbf{q}}{\mathbf{I}}]$ and a new formula for derivatives.

{\bf New formula for $\E[\inner{\mathbf{q}}{\mathbf{I}}]$:} Let $\pi_j$ be the predicate associate to the $j$th statistical query. In other words, $\inner{\mathbf{c}_j}{\mathbf{I}} = |\sigma_{\pi_j}(\mathbf{I})|$. The next lemma applies to any query $\mathbf{q}$ defined by some predicate $\pi$. Recall that $\beta$ is the new variable associated to $\mathbf{q}$ in $P_{\mathbf{q}}$ (Sec.~\ref{subsec:query:answer}).

\begin{lemma}\label{lemma:query_derivative} For any $\ell$ variables $\alpha_{j_1}, \ldots, \alpha_{j_\ell}$ of $P_{\mathbf{q}}$:

  (1) If the logical implication $\pi_{j_1} \wedge \cdots \wedge \pi_{j_\ell} \Rightarrow \pi$ holds, then
  \begin{align}
    \label{eq:aux}
    \frac{\alpha_{j_1}\cdots \alpha_{j_\ell}\partial^\ell P_{\mathbf{q}}}{\partial \alpha_{j_1} \cdots \partial \alpha_{j_\ell}} = &
    \frac{\alpha_{j_1}\cdots \alpha_{j_\ell} \beta \partial^{\ell+1} P_{\mathbf{q}}}{\partial \alpha_{j_1} \cdots \partial \alpha_{j_\ell}\partial \beta}
  \end{align}

  (2) If the logical equivalence $\pi_{j_1} \wedge \cdots \wedge \pi_{j_\ell} \Leftrightarrow \pi$ holds, then
  \begin{align}
    \frac{\alpha_{j_1}\cdots \alpha_{j_\ell}\partial^\ell P_{\mathbf{q}}}{\partial \alpha_{j_1} \cdots \partial \alpha_{j_\ell}} = \frac{\beta \partial P_{\mathbf{q}}}{\partial \beta} \label{eq:aux2}
  \end{align}
\end{lemma}

\begin{proof}
  (1) The proof is immediate by noting that every monomial of $P_{\mathbf{q}}$ that contains all variables $\alpha_{j_1}, \ldots, \alpha_{j_\ell}$ also contains $\beta$; therefore, all monomials on the LHS of Eq.~(\ref{eq:aux}) contain $\beta$ and thus remain unaffected by applying the operator $\beta \partial / \partial \beta$.

  (2) From item (1), we derive Eq.~(\ref{eq:aux}); we prove now that the RHS of Eq.~(\ref{eq:aux}) equals $\frac{\beta \partial P_{\mathbf{q}}}{\partial \beta}$.  We apply item (1) again to the implication $\pi \Rightarrow \pi_{j_1}$ and obtain $\frac{\beta \partial P_{\mathbf{q}}}{\partial \beta} = \frac{\beta \alpha_{j_1} \partial^2 P_{\mathbf{q}}}{\partial \beta \partial \alpha_{j_1}}$ (the role of $\beta$ in Eq.~(\ref{eq:aux}) is now played by $\alpha_{j_1}$). As $P$ is linear, the order of partials does not matter, and this allows us to remove the operator $\alpha_{j_1}\partial/\partial \alpha_{j_1}$ from the RHS of Eq.~(\ref{eq:aux}).  By repeating the argument for $\pi \Rightarrow \pi_{j_2}$, $\pi \Rightarrow \pi_{j_3}$, etc, we remove $\alpha_{j_2}\partial/\partial \alpha_{j_2}$, then $\alpha_{j_3}\partial/\partial \alpha_{j_3}$, etc from the RHS.
\end{proof}

\begin{corollary} (1) Assume $\mathbf{q}$ is defined by a point predicate $\pi = (A_1=v_1 \wedge \cdots \wedge A_\ell=v_\ell)$ for some $\ell \leq m$.  For each $i=1,\ell$, denote $j_i$ the index of the statistic associated to the value $v_i$. In other words, the predicate $\pi_{j_i} \equiv (A_i = v_i)$.  Then,
  \begin{align}
    \label{eq:pqopt}
    \E[\inner{\mathbf{q}}{\mathbf{I}}] = & \frac{n}{P}\frac{\alpha_{j_1}\cdots \alpha_{j_\ell} \partial^\ell P}{\partial \alpha_{j_1} \cdots \partial \alpha_{j_\ell}}
  \end{align}
  (2) Let $\mathbf{q}$ be the query defined by a predicate as in Eq.~(\ref{eq:piq}).  Then,
  \begin{align}
    \label{eq:pqopt2}
    \E[\inner{\mathbf{q}}{\mathbf{I}}] = & \sum_{j_1 \in J_1: \pi_{j_1} \Rightarrow \rho_1} \cdots \sum_{j_m \in J_m: \pi_{j_m} \Rightarrow \rho_m} \frac{n}{P}\frac{\alpha_{j_1}\cdots \alpha_{j_m} \partial^m P}{\partial \alpha_{j_1} \cdots \partial \alpha_{j_m}}
  \end{align}
\end{corollary}

\begin{proof}
(1) Eq.~(\ref{eq:pqopt}) follows from Eq.~(\ref{eq:eq}), Eq.~(\ref{eq:aux2}), and the fact that $P_{\mathbf{q}}[\beta=1] \equiv P$.
(2) Follows from (1) by expanding $\mathbf{q}$ as a sum of point queries as in Lemma.~\ref{lemma:query_derivative} (1).
\end{proof}

In order to compute a query using Eq.~(\ref{eq:pqopt2}), we would have to examine all $m$-dimensional points that satisfy the query's predicate, convert each point into the corresponding 1D statistics, and use Eq.~(\ref{eq:pqopt}) to estimate the count of the number of tuples at this point. Clearly, this is inefficient when $\mathbf{q}$ contains any range predicate containing many point queries.

{\bf New formula for derivatives} Thus, to compute $\E[\inner{\mathbf{q}}{\mathbf{I}}]$, one has to evaluate several partial derivatives of $P$. Recall that $P$ is stored in a highly compressed format, and therefore, computing the derivative may involve nontrivial manipulations. Instead, we use the fact that our polynomial is overcomplete, meaning that $P = \sum_{j \in J_i} \alpha_j P_j$, where $P_j$, $j \in J_i$ does not depend on any variable in $\setof{\alpha_j}{j \in J_i}$ (Eq.~(\ref{eq:p1i})). Let $\rho_i$ be any predicate on the attribute $A_i$.  Then,
\begin{align}
\sum_{j_i \in J_i: \pi_{j_i} \Rightarrow \rho_i} \frac{\alpha_{j_i} \partial P}{\partial \alpha_{j_i}} = & P[\bigwedge_{j \in J_i: \pi_{j_i} \not\Rightarrow \rho_i} \alpha_j = 0]
\end{align}
Thus, in order to compute the summation on the left, it suffices to compute $P$ after setting to $0$ the values of all variables $\alpha_j$, $j \in J_i$ that do not satisfy the predicate $\rho_i$ (this is what the condition $\pi_{j_i} \not\Rightarrow \rho_i$ checks).

Finally, we combine this with Eq.~(\ref{eq:pqopt2}) and obtain the following, much simplified formula for answering a query $\mathbf{q}$, defined by a predicate of the form Eq.~(\ref{eq:piq}):
\begin{align*}
  \E[\inner{\mathbf{q}}{\mathbf{I}}] = \frac{n}{P} P[\bigwedge_{i=1,m} \bigwedge_{j \in J_i: \pi_{j_i} \not\Rightarrow \rho_i} \alpha_j = 0]
\end{align*}
In other words, we set to 0 all 1D variables $\alpha_j$ that correspond to values that do {\em not} satisfy the query, evaluate the polynomial $P$, and multiply it by $\frac{n}{P}$ (which is a precomputed constant independent of the query). For example, if the query ignores an attribute $A_i$, then we leave the 1D variables for that attribute, $\alpha_j$, $j \in J_i$, unchanged. If the query checks a range predicate, $A_i \in [u, v]$, then we set $\alpha_j=0$ for all 1D variables $\alpha_j$ corresponding to values of $A_i$ outside that range.

\begin{example}
Consider three attributes $A$, $B$, and $C$ each with domain 1000 and two multi-dimensional statistics: one $AB$ statistic $A \in [101, 200] \wedge B \in [501, 600]$ and two $BC$ statistics $B \in [551, 650] \wedge C \in [801, 900]$ and $B \in [651, 700] \wedge C \in [701, 800]$.  The polynomial $P$ is shown in Eq.~(\ref{eq:ex:p}).  Consider the query:

\begin{small}
\begin{verbatim}
 q:  SELECT COUNT(*) FROM R 
     WHERE A in [36,150] AND C in [660,834]
\end{verbatim}
\end{small}
We estimate $\mathbf{q}$ using our formula $\frac{n}{P} P[\alpha_{1:35}=0,\ \alpha_{151:1000}=0,\ \gamma_{1:659}=0,\ \gamma_{835:1000}=0]$. There is no need to compute a representation of a new polynomial.
\end{example}

\subsection{Choosing Statistics}
\label{subsec:stat:selection}

In this section, we discuss how we choose the multi-dimensional statistics.  Recall that our summary always includes all 1D statistics of the form $A_i = v$ for all attributes $A_i$ and all values $v$ in the active domain $D_i$. We describe here how to tradeoff the size of the summary for the precision of the MaxEnt model.

A first choice that we make is to include only 2D statistics. It has been shown that restricting to pairwise correlations offers a reasonable compromise between the number of statistics needed and the summary's accuracy~\cite{tzoumas2013efficiently}. Furthermore, we restrict each $A_{i_1}A_{i_2}$ statistic to be a range predicate; i.e., $\pi_j \equiv A_{i_1} \in [u_1,v_1] \wedge A_{i_2} \in [u_2, v_2]$. As explained in Sec.~\ref{subsec:compress}, the predicates over the same attributes $A_{i_1}A_{i_2}$ are disjoint.

The problem is as follows: given a budget $B = B_a*B_s$, which $B_a$ attribute pairs $A_{i_1}A_{i_2}$ do we collect statistics on and which $B_s$ statistics do we collect for each attribute pair? This is a complex problem, and we make the simplifying assumption that $B$ and $B_a$ are known, but we explore different choices of $B_a$ in Sec.~\ref{sec:evaluation}. It is future work to investigate automatic techniques for determining the budgets.

Given $B_a$, there are two main considerations when picking pairs: attribute cover and attribute correlation. If we focus only on correlation, we can pick the set of attribute pairs that are not uniform\footnote{This can be checked by calculating the chi-squared coefficient and seeing if it is close to 0} and have the highest combined correlation such that every pair has at least one attribute not included in any previously chosen, more correlated pair. If we also consider attribute cover, we can pick the set of pairs that covers the most attributes with the highest combined correlation. For example, if $B_a = 2$ and we have the attribute pairs $BC$, $AB$, $CD$, and $AD$ in order of most to least correlated, if we only consider correlation, we would choose $AB$ and $BC$. However, if we consider attribute cover, we would choose $AB$ and $CD$. We experiment with both of these choices in Sec.~\ref{sec:evaluation}, and, in the end, conclude that considering attribute cover achieves more precise query results for the same budget than the alternative.

Next, for a given attribute pair of type $A_{i_1}A_{i_2}$, we need to choose the best $B_s$ 2D range predicates $[u_1,v_1]\times[u_2,v_2]$. We consider three heuristics and show experimental results to determine which technique yields, on average, the lowest error on query results.

{\bf LARGE SINGLE CELL} In this heuristic, the range predicates are single point predicates, $A_{i_1}=u_1 \wedge A_{i_2}=u_2$, and we choose the points $(u_1,u_2)$ as the $B_s$ most popular values in the two dimensional space; i.e., the $B_s$ largest values of $|\sigma_{A_{i_1}=u_1 \wedge A_{i_2}=u_2}(I)|$.

{\bf ZERO SINGLE CELL} In this heuristic, we select the empty/zero/nonexistent cells; i.e., we choose $B_s$ points $(u_1,u_2)$ s.t. $\sigma_{A_{i_1}=u_1 \wedge A_{i_2}=u_2}(I)=\emptyset$. If there are fewer than $B_s$ such points, we choose the remaining points as in SINGEL CELL. The justification for this heuristic is that, given only the 1D statistics, the MaxEnt model will produce false positives (``phantom'' tuples) in empty cells; this is the opposite problem encountered by sampling techniques, which return false negatives. This heuristic has another advantage because the value of $\alpha_j$ in $P$ is always 0 and does not need to be updated during solving.

{\bf COMPOSITE} This method partitions the entire space $D_{i_1} \times D_{i_2}$ into a set of $B_s$ disjoint rectangles and associates one statistic with each rectangle. We do this using an adaptation of KD-trees.

The only difference between our KD-tree algorithm and the traditional one is our splitting condition. Instead of splitting on the median, we split on the value that has the lowest sum squared average value difference. This is because we want our KD-tree to best represent the true values. Suppose we have cell counts on dimensions $A$ and $A'$ as shown in Fig.~\ref{fig:kd_methodcompare} (a). For the next vertical split, if we followed the standard KD-tree algorithm, we would choose the second split. Instead, our method chooses the first split. Using the first split minimizes the sum squared error.

\begin{figure}[!t]
  \begin{minipage}[c]{.18\textwidth}
  \centering
  \subfloat[]{\includegraphics[width=\textwidth]{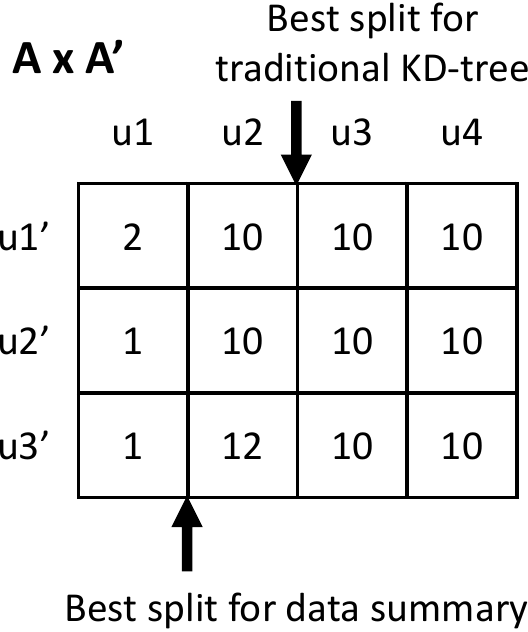}}
  \end{minipage}%
  \begin{minipage}[c]{.32\textwidth}
  \centering
  \subfloat[]{\includegraphics[width=\textwidth]{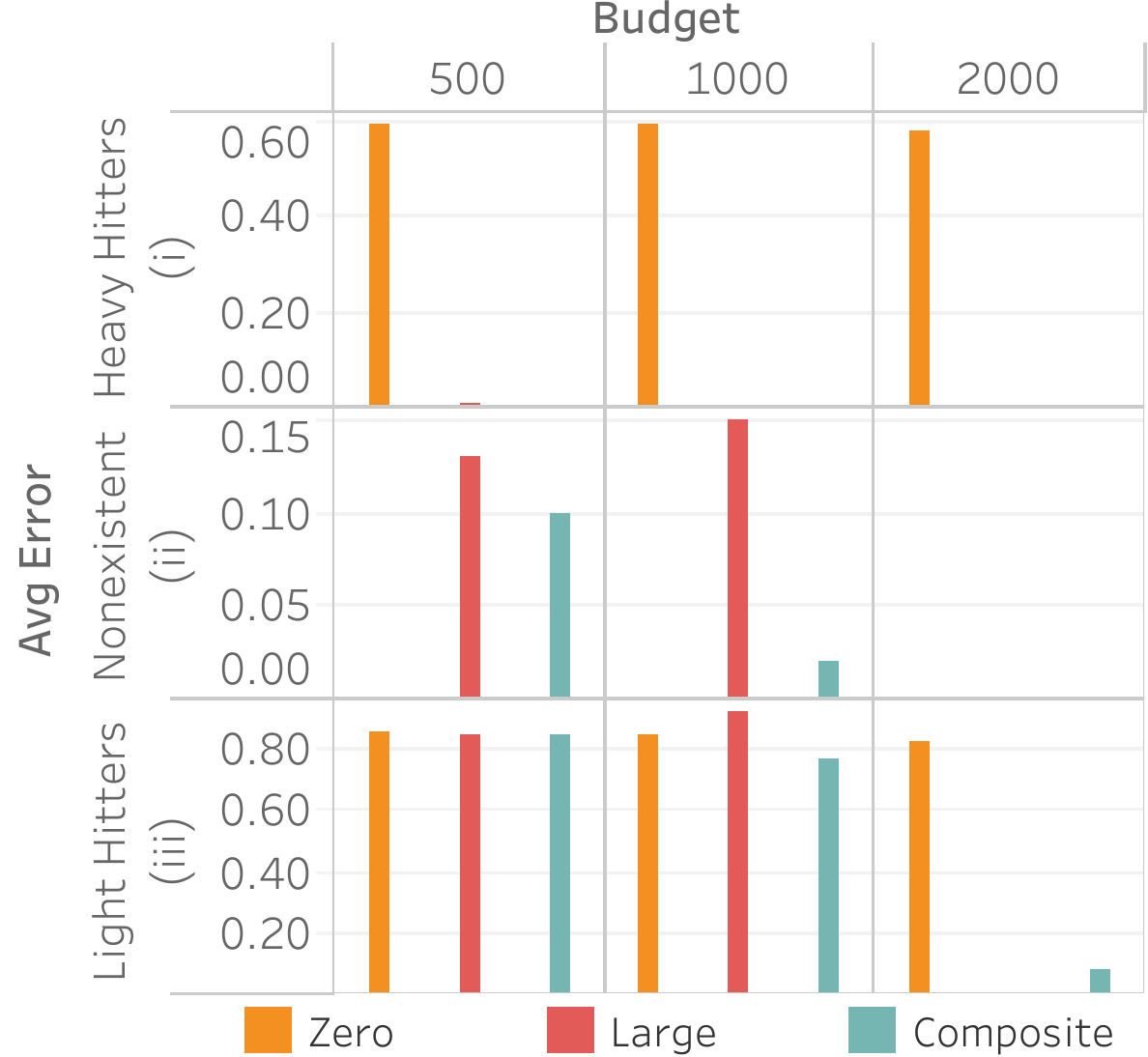}}
  \end{minipage}
  \caption{(a) Illustration of the splitting step, and (b) query accuracy versus budget for the three different heuristics and three different selections: (b.i) selecting 100 heavy hitter values, (b.ii) selecting nonexistent values, and (b.iii) selecting 100 light hitter values.}
  \label{fig:kd_methodcompare}
\end{figure}

Our COMPOSITE method repeatedly splits the attribute domains $D_{i_1}$ and $D_{i_2}$ (alternating) by choosing the best split value until it exhausts the budget $B_s$. Then, for each rectangle $[u_1,v_1] \times [u_2,v_2]$ in the resulting KD-tree, it creates a 2D statistic $(\mathbf{c}_j, s_j)$, where the query $\mathbf{c}_j$ is associated with the 2D range predicate and the numerical value $s_j \eqdef |\sigma_{A_{i_1} \in [u_1,v_1] \wedge A_{i_2} \in [u_2,v_2]}(I)|$.

We evaluate the three heuristics on the dataset of flights in the United States restricted to the attributes (fl\texttt{\_}date, fl\texttt{\_}time, distance) (see Sec \ref{sec:evaluation} for details). We gather statistics using the different techniques and different budgets on the attribute pair (fl\texttt{\_}time, distance). There are 5,022 possible 2D statistics, 1,334 of which exist in \texttt{Flights}. We evaluate the accuracy of the resultant count of the query \texttt{SELECT fl\_time, distance, COUNT(*) FROM Flights WHERE fl\_time = x AND distance = y GROUP BY fl\_time, distance} for 100 heavy hitter (x, y) values, 100 light hitter (x, y) values, and 200 random (x, y) nonexistent/zero values. We choose 200 zero values to match the 100+100 heavy and light hitters.

Figure \ref{fig:kd_methodcompare} (b.i) plots the query accuracy versus method and budget for 100 heavy hitter values. Both {\bf LARGE} and {\bf COMPOSITE} achieve almost zero error for the larger budgets while {\bf ZERO} gets around 60 percent error no matter the budget.

(b.ii) plots the same for nonexistent values, and clearly {\bf ZERO} does best because it captures the zero values first. {\bf COMPOSITE}, however, gets a low error with a budget of 1,000 and outperforms {\bf LARGE}. Interestingly, {\bf LARGE} does slightly worse with a budget of 1,000 than 500. This is a result of the final value of $P$ being larger with a larger budget, and this makes our estimates slightly higher than 0.5, which we round up to 1. With a budget of 500, our estimates are slightly lower than 0.5, which we round down to 0.

Lastly, (b.iii) plots the same for 100 light hitter values, and while {\bf LARGE} eventually outperforms {\bf COMPOSITE}, {\bf COMPOSITE} gets similar error for all budgets. In fact, {\bf COMPOSITE} outperforms {\bf LARGE} for a budget of 1,000 because {\bf LARGE} predicts that more of the light hitter values are nonexistent than it does with a smaller budget as less weight is distributed to the light hitter values.

Overall, we see that {\bf COMPOSITE} is the best method to use across all queries and is the technique we use in our evaluation. Note that when the budget is high enough, around 5,000, all methods get almost no error on all queries because we have a large enough budget to capture the entire domain. Also note that the uncompressed polynomial size of these summarizes are orders of magnitude larger than the compressed versions. For example, for a budget of 2,000, the uncompressed polynomial has 4.4 million terms while the compressed polynomial has only 9,000 terms.

\section{Implementation}
\label{sec:implementation}
We implemented our polynomial solver and query evaluator in Java 1.8, in a prototype system that we call \name. We created our own polynomial class and variable types to implement our factorization. Since the initial implementation of our solver took an estimated 3 months to run on the experiments in Sec \ref{sec:evaluation}, we further optimized our factorization and solver by using bitmaps to associate variables with their statistics and using Java's parallel streaming library to parallelize polynomial evaluation. Also, by using Java to answer queries, we were able to store the polynomial factorization in memory. By utilizing these techniques, we reduced the time to learn the model (solver runtime) to 1 day and saw a decrease in query answering runtime from around 10 sec to 500 ms (95\% decrease). We show more query evaluation performance results in Sec.~\ref{sec:evaluation}.

Lastly, we stored the polynomial variables in a Postgres 9.5.5 database and stored the polynomial factorization in a text file. We perform all experiments on a 64bit Linux machine running Ubuntu 5.4.0. The machine has 120 CPUs and 1 TB of memory. For the timing results, the Postgres database, which stores all the samples, also resides on this machine and has a shared buffer size of 250 GB.

\section{Evaluation}
\label{sec:evaluation}
In this section, we evaluate the performance of \name in terms of query accuracy and query execution time. We compare our approach to uniform sampling and stratified sampling.

\subsection{Experimental Setup}
For all our summaries, we ran our solver for 30 iterations using the method presented in Sec.~\ref{subsec:solving} or until the error was below $1 \times 10^{-6}$. Our summaries took under 1 day to compute with the majority of the time spent building the polynomial and solving for the parameters.

We evaluate \name on two real datasets as opposed to benchmark data to measure query accuracy in the presence of naturally occurring attribute correlations. The first dataset comprises information on flights in the United States from January 1990 to July 2015~\cite{rita}. We load the data into PostgreSQL, remove null values, and bin all real-valued attributes into equi-width buckets. We further reduce the size of the active domain to decrease memory usage and solver runtime by binning cities such that the two most popular cities in each state are separated and the remaining less popular cities are grouped into a city called `Other'. We use equi-width buckets to facilitate transforming a user's query into our domain and to avoid hiding outliers, but it is future work to try different bucketization strategies. The resulting relation, \texttt{FlightsFine(fl\_date, origin\_city, dest\_city, fl\_time, distance)}, is 5 GB in size.

To vary the size of our active domain, we also create \texttt{FlightsCoarse(fl\_date, origin\_state, dest\_state, fl\_time, distance)}, where we use the origin state and destination state as flight locations. The left table in Fig.~\ref{fig:attr_size} shows the resulting active domain sizes. 

The second dataset is 210 GB in size. It comprises N-body particle simulation data~\cite{ChaNGaScaling}, which captures the state of astronomy simulation particles at different moments in time (snapshots). The relation \texttt{Particles(density, mass, x, y, z, grp, type, snapshot)} contains attributes that capture particle properties and a binary attribute, grp, indicating if a particle is in a cluster or not. We bucketize the continuous attributes (density, mass, and position coordinates) into equi-width bins. The right table in Fig.~\ref{fig:attr_size} shows the resulting domain sizes.

\begin{figure}
    \scriptsize
    \centering
    \begin{tabular}{|M{39pt}|M{33pt}|M{33pt}|}
    \hline
    & \texttt{Flights} \texttt{Coarse} & \texttt{Flights} \texttt{Fine} \\ \hline
    \texttt{fl\_date (FD)} & 307 & 307 \\ \hline
    \texttt{origin (OS/OC)} & 54 & 147 \\ \hline
    \texttt{dest (DS/DC)} & 54 & 147 \\ \hline
    \texttt{fl\_time (ET)} & 62 & 62 \\ \hline
    \texttt{distance (DT)} & 81 & 81 \\ \hline
    \# possible tuples & $4.5 \times 10^9$ & $3.3 \times 10^{10}$ \\ \hline
   \end{tabular}
    \begin{tabular}{|c|c|}
    \hline
    & \texttt{Particles} \\ \hline
    \texttt{density} & 58 \\ \hline
    \texttt{mass} & 52 \\ \hline
    \texttt{x} & 21 \\ \hline
    \texttt{y} & 21 \\ \hline
    \texttt{z} & 21 \\ \hline
    \texttt{grp} & 2 \\ \hline
    \texttt{type} & 3 \\ \hline
    \texttt{snapshot} & 3 \\ \hline
    \# possible & \\
    tuples & $5.0 \times 10^8$ \\ \hline
    \end{tabular}
    \caption{Active domain sizes. Each cell shows the number of distinct values after binning. Abbreviations shown in brackets are used in figures to refer to attribute names: e.g., OS stands for origin\texttt{\_}state. }
    \label{fig:attr_size}
\end{figure}

\subsection{Query Accuracy}
\label{subsec:results:accuracy}

We first compare \name to uniform and stratified sampling on the flights dataset. We use one percent samples, which require approximately 100 MB of space when stored in PostgreSQL. To approximately match the sample size, our largest summary requires only 600 KB of space in PostgreSQL to store the polynomial variables and approximately 200 MB of space in a text file to store the polynomial factorization. This, however, could be improved and factorized further beyond what we did in our prototype implementation.

We compute correlations on \texttt{FlightsCoarse} across all attribute-pairs and identify the following pairs as having the largest correlations (C stands for ``coarse''): 1C = (origin\texttt{\_}state, distance), 2C = (destination\texttt{\_}state, distance), 3 = (fl\texttt{\_}time, distance), and 4C = (origin\texttt{\_}state, destination\texttt{\_}state). We use the corresponding attributes, which are also the most correlated, for the finer-grained relation and refer to those attribute-pairs as 1F, 2F, and 4F.

As explained in Sec.~\ref{subsec:stat:selection}, we build four summaries with a budget $B = 3,000$, chosen to keep runtime under a day while allowing for variations of $B_a$ (``breadth'') and $B_s$ (``depth''), to show the difference in choosing statistics based solely on correlation (choosing statistics in order of most to least correlated) versus attribute cover (choosing statistics that cover the attributes with the highest combined correlation). The first summary, No2D, contains only 1D statistics. The next two, Ent1\&2 and Ent3\&4, use 1,500 buckets across the attribute-pairs (1, 2) and (3, 4), respectively. The final one, Ent1\&2\&3, uses 1,000 buckets for the three attribute-pairs (1, 2, 3). We do not include 2D statistics related to the flight date attribute because this attribute is relatively uniformly distributed and does not need a 2D statistic to correct for the MaxEnt's underlying uniformity assumption. Fig~\ref{fig:method_summary} summarizes the summaries.



For sampling, we choose to compare with a uniform sample and four different stratified samples. We choose the stratified samples to be along the same attribute-pairs as the 2D statistics in our summaries; i.e., pair 1 through pair 4.

\begin{figure}[t]
    \scriptsize
    \centering
    \begin{tabular}{|c|c|c|c|c|c|}
    \hline
    & MaxEnt Method & No2D & 1\&2 & 3\&4 & 1\&2\&3 \\ \hline
    Pair 1 & (origin, distance) & & X & & X \\ \hline
    Pair 2 & (dest, distance) & & X & & X \\ \hline
    Pair 3 & (time, distance) & & & X & X \\ \hline
    Pair 4 & (origin, dest) & & & X & \\ \hline
    \end{tabular}
    \caption{MaxEnt 2D statistics including in the summaries. The top row is the label of the MaxEnt method used in the graphs.}
    \label{fig:method_summary}
\end{figure}

To test query accuracy, we use the following query template:

\begin{scriptsize}
\begin{lstlisting}
    SELECT A1,..., Am COUNT(*)
    FROM R WHERE A1=`v1' AND ... AND Am=`vm';
\end{lstlisting}
\end{scriptsize} 

We test the approaches on 400 unique \texttt{(A1,.., Am)} values. We choose the attributes for the queries in a way that illustrates the strengths and weaknesses of \name. For the selected attributes, 100 of the values used in the experiments have the largest count (heavy hitters), 100 have the smallest count (light hitters), and 200 (to match the 200 existing values) have a zero true count (nonexistent/null values). To evaluate the accuracy of \name, we compute $|true-est|/(true+est)$ on the heavy and light hitters. To evaluate how well \name distinguishes between rare and nonexistent values, we compute the F measure, $2*\textrm{precision}*\textrm{recall}/(\textrm{precision}+\textrm{recall})$ with precision $= |\{est_t > 0\ :\ t \in \textrm{light hitters}\}|/|\{est_t > 0\ :\ t \in (\textrm{light hitters}\ \cup\ \textrm{null values})\}|$ and recall $= |\{est_t > 0\ :\ t \in \textrm{light hitters}\}|/100$. We do not compare the runtime of \name to sampling for the flights data because the dataset is small, and the runtime of \name is, on average, below 0.5 seconds and at most 1 sec. Sec.~\ref{subsec:results:scalability} reports runtime for the larger data.

\begin{figure}[t]
    \centering
    \includegraphics[width=0.48\textwidth]{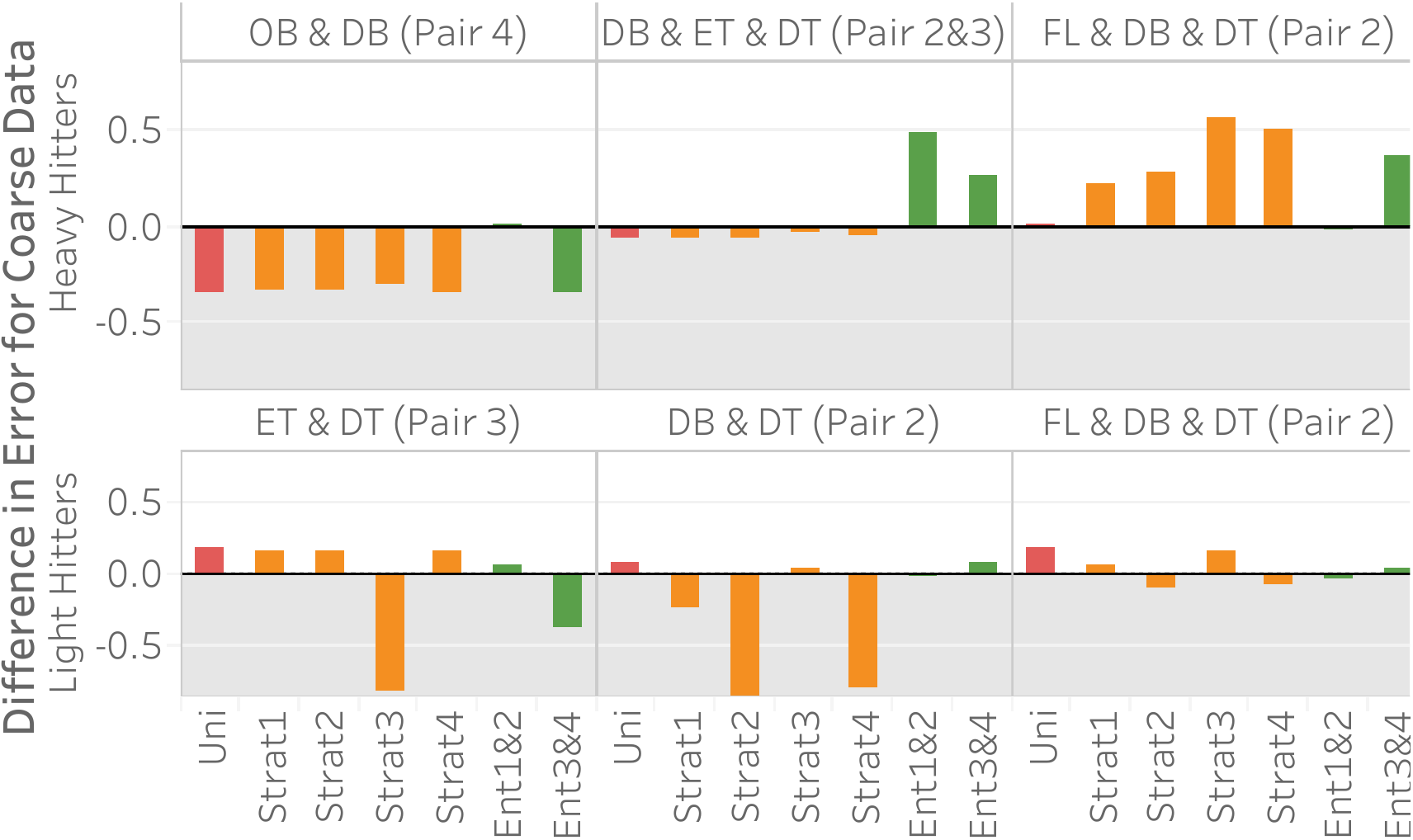}
    \caption{Query error difference between all methods and Ent1\&2\&3 over \texttt{FlightsCoarse}. The pair in parenthesis in the column header corresponds to the 2D statistic pair(s) used in the query template. For reference, pair 1 is (origin/OB, distance/DT), pair 2 is (dest/DB, distance/DT), pair 3 is (time/ET, distance/DT), and pair 4 is (origin/OB, dest/DB).}
    \label{fig:coarsequeries}
\end{figure}

Fig.~\ref{fig:coarsequeries} (top) shows query error differences between all methods and Ent1\&2\&3 (i.e., average error for method X minus average error for Ent1\&2\&3) for three different heavy hitter queries over \texttt{FlightsCoarse}. Hence, bars above zero indicate that Ent1\&2\&3 performs better and vice versa. Each of the three query templates uses a different set of attributes that we manually select to illustrate different scenarios. The attributes of the query are shown in the column header in the figure, and any 2D statistic attribute-pair contained in the query attributes is in parentheses. Each bar shows the average of 100 query instances selecting different values for each template. 

As the figure shows, Ent1\&2\&3 is comparable or better than sampling on two of the three queries and does worse than sampling on query 1. The reason it does worse on query 1 is that it does not have any 2D statistics over 4C, the attribute-pair used in the query, and 4C is fairly correlated. Our lack of a 2D statistic over 4C means we cannot correct for the MaxEnt's uniformity assumption. On the other hand, all samples are able to capture the correlation because the 100 heavy hitters for query 1 are responsible for approximately 25\% of the data. This is further shown by Ent3\&4, which has 4C as one of its 2D statistics, doing better than Ent1\&2\&3 on query 1.

Ent1\&2\&3 is comparable to sampling on query 2 because two of its 2D statistics cover the three attributes in the query. It is better than both Ent1\&2 and Ent3\&4 because each of those methods has only one 2D statistic over the attributes in the query. Finally, Ent1\&2\&3 is better than stratified sampling on query 3 because it not only contains a 2D statistic over 2C but also correctly captures the uniformity of flight date. This uniformity is also why Ent1\&2 and a uniform sample do well on query 3. Another reason stratified sampling performs poorly on query 3 is because the result is highly skewed in the attributes of destination state and distance but remains uniform in flight date. The top 100 heavy hitter tuples all have the destination of `CA' with a distance of 300. This means even a stratified sample over destination state and distance will likely not be able to capture the uniformity of flight date within the strata for `CA' and 300 miles.

Fig.~\ref{fig:coarsequeries} (bottom) shows results for the same queries but for the bottom 100 light hitter values. In this case, \name always does better than uniform sampling. Our performance compared to stratified sampling depends on the stratification and query. Stratified sampling outperforms Ent1\&2\&3 when the stratification is exactly along the attributes involved in the query. For example, for query 1, the sample stratified on pair 3 outperforms \name by a significant amount because pair 3 is computed along the attributes in query 1. Interestingly, Ent3\&4 and Ent1\&2 do better than Ent1\&2\&3 on query 1 and query 2, respectively. Even though both of the query attributes for query 1 and query 2 are statistics in Ent1\&2\&3, Ent1\&2 and Ent3\&4 have more buckets and are thus able to capture more zero elements. Lastly, we see that for query 3, we are comparable to stratified sampling because we have a 2D statistic over pair 2C, and the other attribute, flight date, is relatively uniformly distributed in the query result.

We ran the same queries over the \texttt{FlightsFine} dataset and found \textit{identical} trends in error difference. We omit the graph due to space constraints.

\begin{figure}[t]
    \centering
    \includegraphics[width=0.98\linewidth]{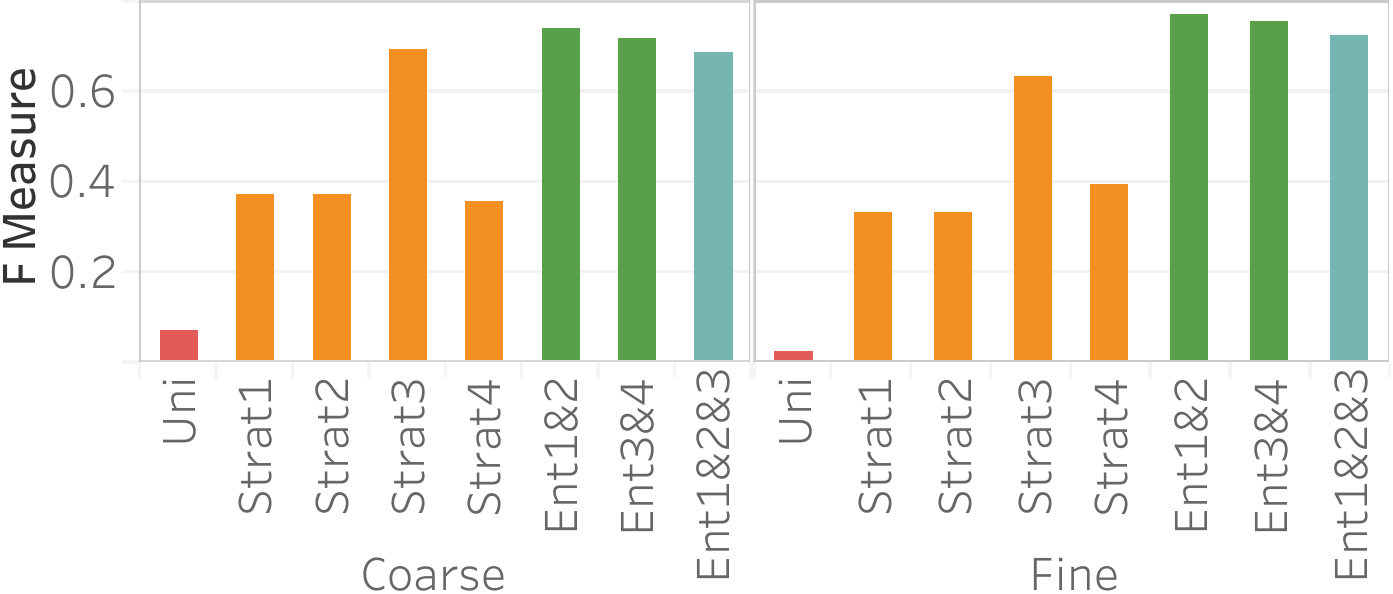}
    \caption{F measure for light hitters and null values over \texttt{FlightsCoarse} (left) and \texttt{FlightsFine} (right).}
    \label{fig:fmeasure}
\end{figure}

An important advantage of our approach is that it more accurately distinguishes between rare values and nonexistent values compared with stratified sampling, which often does not have samples for rare values when the stratification does not match the query attributes. To assess how well our approach works on those rare values, Fig.~\ref{fig:fmeasure} shows the average F measure over fifteen 2- and 3-dimensional queries selecting light hitters and null values.

We see that Ent1\&2 and 3\&4 have F measures close to 0.72, beating all stratified samples and also beating Ent1\&2\&3. The key reason why they beat Ent1\&2\&3 is that these summaries have the largest numbers of buckets, which ensures they have more fine grained information and can more easily identify regions without tuples. Ent1\&2\&3 has an F measure close to 0.69, which is slightly lower than the stratified sample over pair 3 but better than all other samples. The reason the sample stratified over pair 3 performs well is that the flight time attribute has a more skewed distribution and has more rare values than other dimensions. A stratified sample over that dimensions will be able to capture this. On the other hand, Ent1\&2\&3 will estimate a small count for any tuple containing a rare flight time value and will be rounded to 0.

\subsection{Scalability}
\label{subsec:results:scalability}

To measure the performance of \name on large-scale datasets, we use three subsets of the 210 GB \texttt{Particles} table. We select data for one, two, or all three snapshots (each snapshot is approximately 70 GB in size). We build a 1 GB uniform sample for each subset of the table as well as a stratified sample over the pair density and group with the same sampling percentage as the uniform sample. We then build two MaxEnt summaries; EntNo2D uses no 2D statistics, and EntAll contains 5 2D statistics with 100 buckets over each of the most correlated attributes, not including snapshot. We run a variety of 4D selection queries such as the ones from Sec.~\ref{subsec:results:accuracy}, split into heavy hitters and light hitters. We record the query accuracy and runtime.

\begin{figure}[t]
    \centering
    \includegraphics[width=0.98\linewidth]{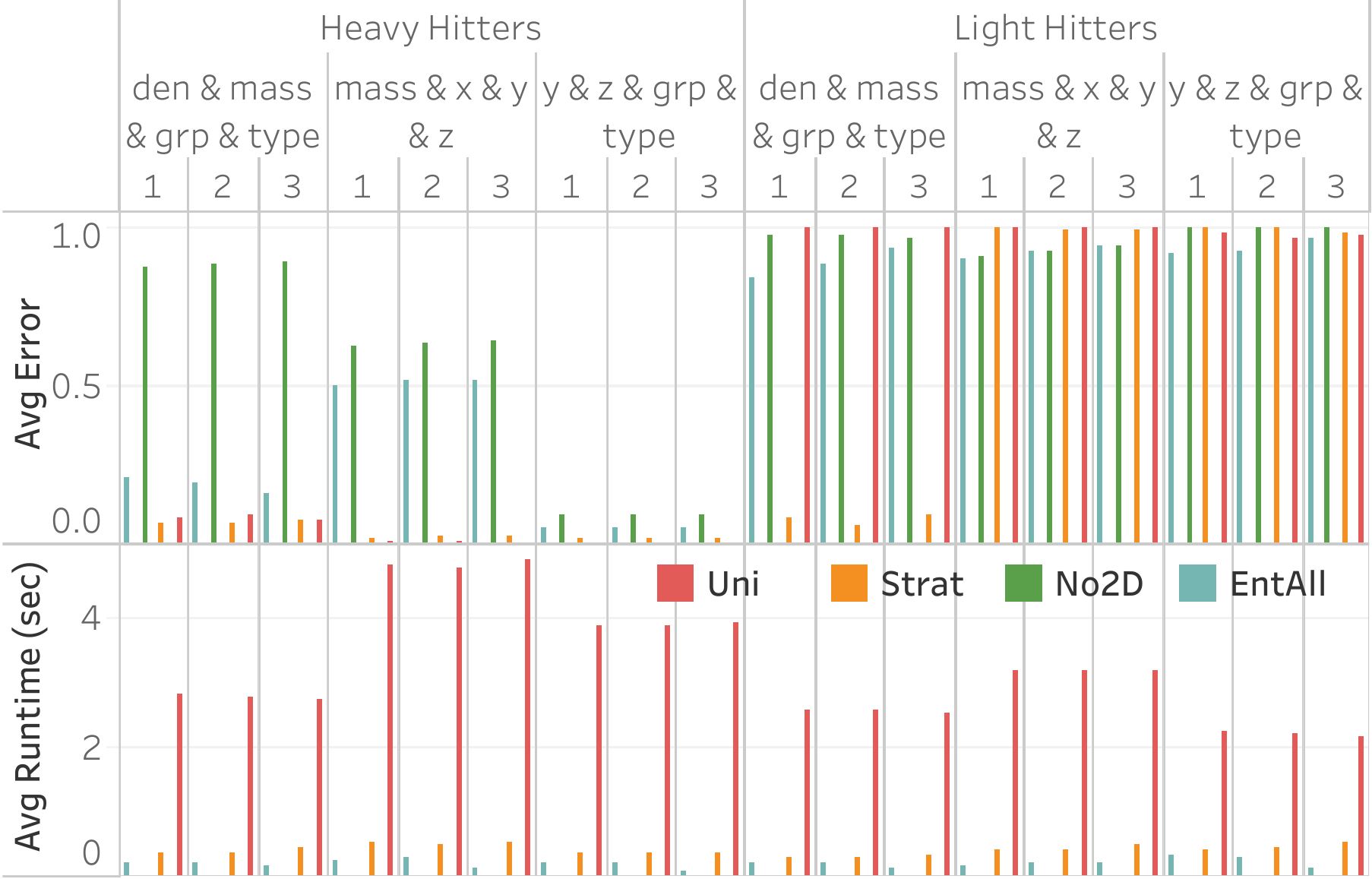}
    \caption{Query average error and runtime for three 4D selection queries on the \texttt{Particles} table. The stratified sample (orange) is stratified on (den, grp).}
    \label{fig:vulcan}
\end{figure}

Fig.~\ref{fig:vulcan} shows the query accuracy and runtime for three different selection queries as the number of snapshots increases. We see that \name consistently does better than sampling on query runtime, although both \name and stratified sampling execute queries in under one second. Stratified sampling outperforms uniform sampling because the stratified samples are generally smaller than their equally selective uniform sample.

In terms of query accuracy, sampling always does better than \name for the heavy hitter queries. This is expected because the bucketization of \texttt{Particles} is relatively coarse grained, and a 1 GB sample is sufficiently large to capture the heavy hitters. We do see that EntAll does significantly better than EntNo2D for query 1 because three of its five statistics are over the attributes of query 1 while only 1 statistic is over the attributes of queries 2 and 3. However, the query results of query 3 are more uniform, which is why EntNo2D and EntAll do well. 

For the light hitter queries, none of the methods do well except for the stratified sample in query 1 because the query is over the attributes used in the stratification. EntAll does slightly better than stratified sampling on queries 2 and 3.
\subsection{Statistics Selection}
\label{subsec:results:statistic_choice}

To compare different 2D statistic choices for our method, we look at the query accuracy of the four different MaxEnt methods summarized in Fig.~\ref{fig:maxentcompare}. We use the flights dataset and query templates from Sec.~\ref{subsec:results:accuracy}. We run six different two-attribute selection queries over all possible pairs of the attributes covered by pair 1 through 4; i.e., origin, destination, time, and distance. We select 100 heavy hitters, 100 light hitters, and 200 null values.

\begin{figure}[!t]
  \centering
  \subfloat[]{\includegraphics[width=0.23\textwidth]{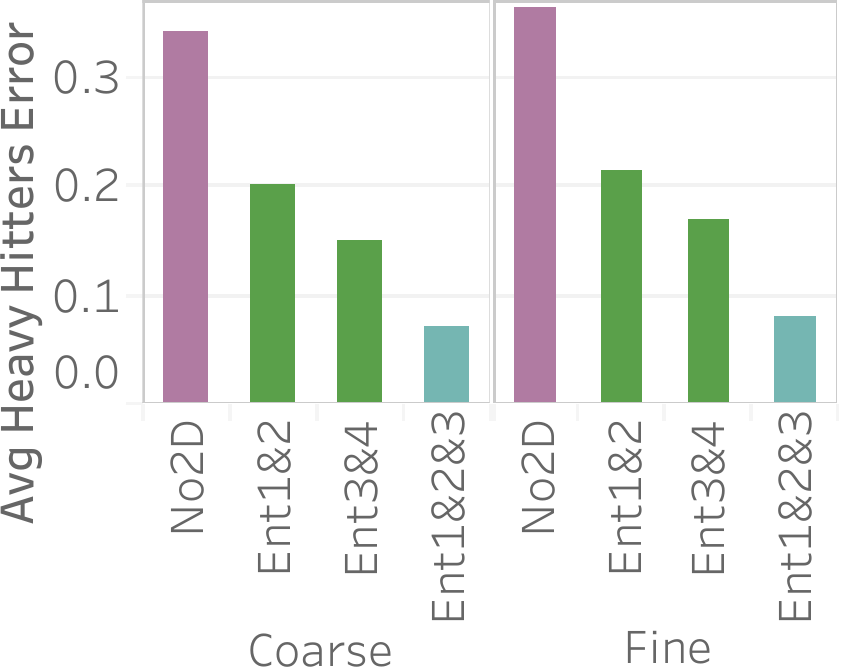}}
  \hfill
  \subfloat[]{\includegraphics[width=0.23\textwidth]{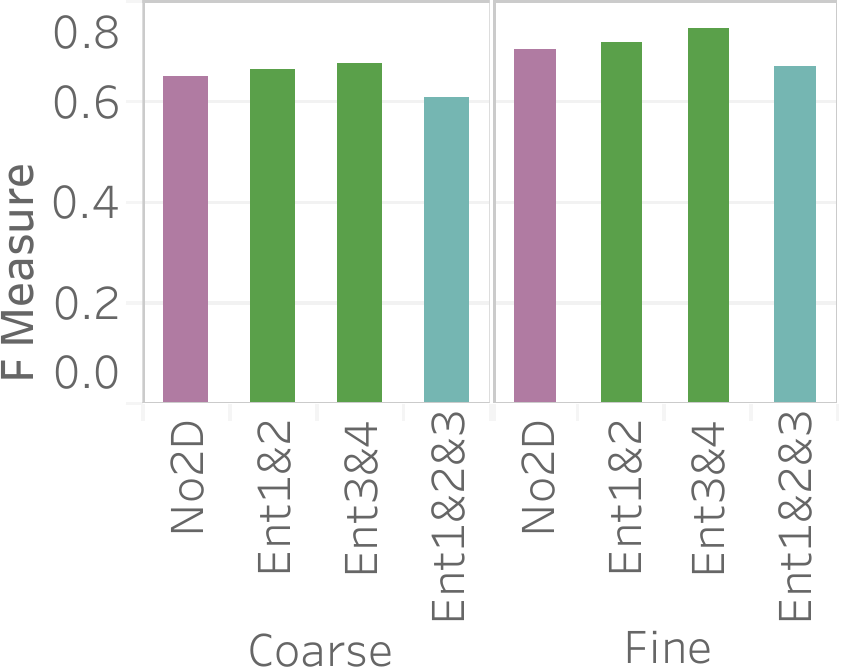}}
  \caption{(a) Error over 2D heavy hitter queries and (b) F measure over 2D light hitter and null value queries across different MaxEnt methods over \texttt{FlightsCoarse} and \texttt{FlightsFine}.}
  \label{fig:maxentcompare}
\end{figure}

Fig.~\ref{fig:maxentcompare} shows the average error for the heavy hitters and the F measure for light hitters across the six queries. Overall, we see that the summary with more attribute-pairs but fewer buckets (more ``breadth''), Ent1\&2\&3, does best on the heavy hitters. On the other hand, for the light hitters, we see that the summary with fewer attribute-pairs but more buckets (more ``depth'') and still covers the attributes, Ent3\&4, does best. Ent3\&4 doing better than Ent1\&2 implies that choosing the attribute-pairs that cover the attributes yields better accuracy than choosing the most correlated pairs because even though Ent1\&2 has the most correlated attribute-pairs, it does not have a statistic containing flight time. Lastly, Ent1\&2\&3 does best on the heavy hitter queries yet slightly worse on the light hitter queries because it does not have as many buckets as Ent1\&2 and Ent3\&4 and can thus not capture as many regions in the active domains with no tuples.

\section{Discussion}
\label{sec:discussion}
The above evaluation shows that \name is competitive with stratified sampling overall and better at distinguishing between infrequent and absent values. Importantly, unlike stratified sampling, \name's summaries permit multiple 2D statistics. The main limitations of \name are the dependence on the size of the active domain, correlation-based 2D statistic selection, manual bucketization, and limited query support.

To address the first problem, our future work is to investigate using standard algebraic factorization techniques on non-materializable polynomials. By further reducing the polynomial size, we will be able to handle larger domain sizes. We also will explore using statistical model techniques to more effectively decompose the attributes into 2D pairs, similar to \cite{deshpande2001independence}. To no longer require bucketizing categorical variables (like city), we will research hierarchical polynomials. These polynomials will start with coarse buckets (like states), and build separate polynomials for buckets that require more detail. This may require the user to wait while a new polynomial is being loaded but would allow for different levels of query accuracy without sacrificing polynomial size.

Lastly, to address our queries not reporting error, we will add variance calculations to query answers. We have a closed-form formula for calculating variance for a single statistic but still need to expand the formula to handle more complex queries. Additionally, our theoretical model can support more complex queries involving joins and other aggregates, but it is future work to implement these queries and make them run efficiently.

\section{Related Work}
\label{sec:related_work}
Although there has been work in the theoretical aspects of probabilistic databases \cite{suciu2011probabilistic}, as far as we could find, there is not existing work on using a probabilistic database for data summarization. However, there has been work by Markl \cite{markl2005consistently} on using the maximum entropy principle to estimate the selectivity of predicates. This is similar to our approach except we are allowing for multiple predicates on an attribute and are using the results to estimate the likelihood of a tuple being in the result of a query rather than the likelihood of a tuple being in the database.

Our work is also similar to that by Suciu and R\'{e} \cite{re2012understanding} except their goal was to estimate the size of the result of a query rather than tuple likelihood. Their method also relied on statistics on the number of distinct values of an attribute whereas our statistics are based on the selectivity of each value of an attribute. 

There has been much research in sampling techniques for faster query processing. In the work by Chaudhiri et al. \cite{chaudhuri2001robust}, they precompute the samples of data that minimizes the errors due to variance in the data for a specific set of queries they predict. The work by \cite{babcock2003dynamic} chooses multiple samples to use in query execution but only considers single column stratifications. The work by \cite{ding2016samplelus} builds a measure-biased sample for each measure dimension to handle sum queries and uniform samples to handle count queries. Depending on if the query is highly selective or not, they choose an appropriate sample. The later work of BlinkDB \cite{agarwal2013blinkdb} improves this by removing assumptions on the queries. BlinkDB only assumes that there is a set of columns that are queried, but the values for these columns can be anything among the possible set of values. BlinkDB then computes samples for each possible value of the predicate column in an online fashion and chooses the single best sample to run when a user executes a query.

Although we handle linear queries, our work makes no assumptions on query workload and can take into account multi-attribute combinations when choosing statistics. When a user executes a query, our method does not need to choose which summary to use. Further, the summary building is all done offline.

\section{Conclusion}
\label{sec:conclusion}
We presented, \name, a new approach to generate probabilistic database summaries for interactive data exploration using the Principle of Maximum Entropy. Our approach is complementary to sampling. Unlike sampling, \name's summaries strive to be independent of user queries and capture correlations between multiple different attributes at the same time. Results from our prototype implementation on two real-world datasets up to 210 GB in size demonstrate that this approach is competitive with sampling for queries over frequent items while outperforming sampling on queries over less common items.

\begin{small}
\noindent\textbf{Acknowledgments} This work is supported by NSF 1614738 and NSF 1535565. Laurel Orr is supported by the NSF Graduate Research Fellowship.
\end{small}

\end{sloppypar}
\small
\bibliographystyle{abbrv}
\bibliography{references}
\end{document}